%% file: main.tex
\documentclass[onecolumn,a4paper]{IEEEtran}
\IEEEoverridecommandlockouts%
\usepackage{graphicx}      

\usepackage{amsmath,amssymb,amsfonts}
\usepackage{mathtools}
\usepackage{amsthm}
\usepackage{algorithmic}
\usepackage{textcomp}
\usepackage{tikz}
\usepackage{caption}
\usepackage{cuted}
\usepackage{pgfgantt}
\usetikzlibrary{arrows,positioning,calc,intersections}
\usetikzlibrary{datavisualization.formats.functions}
\usepackage{pgfplots}
\usepgfplotslibrary{fillbetween}
\usetikzlibrary{arrows, decorations.markings}
\usetikzlibrary{arrows.meta}

\newtheorem{Theorem}{Theorem}
\newtheorem*{Theorem*}{Theorem}
\newtheorem{Lemma}[Theorem]{Lemma}

\newtheorem{Definition}[Theorem]{Definition}

\newtheorem{Corollary}[Theorem]{Corollary}
\newtheorem{Remark}[Theorem]{Remark}

\newcommand{\up}{u^\prime}
\newcommand{\upp}{u^{\prime\prime}}
\newcommand{\Mp}{M^\prime}
\newcommand{\Mpp}{M^{\prime\prime}}
\newcommand{\dc}{\mathcal{D}}
\newcommand{\dcp}{\mathcal{D}^\prime}
\newcommand{\dcpp}{\mathcal{D}^{\prime\prime}}
\newcommand{\C}{\mathcal{C}}
\newcommand{\Cp}{\mathcal{C}^\prime}
\newcommand{\Cpp}{\mathcal{C}^{\prime\prime}}

\newcommand{\norm}[1]{\left\lVert#1\right\rVert}

\newcommand{\setx}{\ensuremath{\mathcal{X}}}
\newcommand{\sety}{\ensuremath{\mathcal{Y}}}

\newcommand{\setd}{\ensuremath{\mathcal{D}}}

\newcommand{\sete}{\ensuremath{\mathcal{E}}}

\newcommand{\setn}{\ensuremath{\mathcal{N}}}
\newcommand{\setp}{\ensuremath{\mathcal{P}}}
\newcommand{\setf}{\ensuremath{\mathcal{F}}}
\newcommand{\sett}{\ensuremath{\mathcal{T}}}

\newcommand{\sets}{\ensuremath{\mathcal{S}}}

\begin{document}

\title{Joint Identification and Sensing for Discrete Memoryless Channels}

\author{
\IEEEauthorblockN{Wafa Labidi \IEEEauthorrefmark{1}\IEEEauthorrefmark{2}\IEEEauthorrefmark{4}\thanks{\IEEEauthorrefmark{4}BMBF Research Hub 6G-life, Germany}, Yaning Zhao\IEEEauthorrefmark{1}\IEEEauthorrefmark{2}, Christian Deppe\IEEEauthorrefmark{2}\IEEEauthorrefmark{4}, Holger Boche\IEEEauthorrefmark{1}\thanks{\IEEEauthorrefmark{3}Cyber Security in the Age of Large-Scale Adversaries–
Exzellenzcluster, Ruhr-Universit\"at Bochum, Germany}\IEEEauthorrefmark{4}\thanks{\IEEEauthorrefmark{6}{\color{black}{ Munich Center for Quantum Science and Technology (MCQST) }}}\IEEEauthorrefmark{7}\thanks{\IEEEauthorrefmark{7} {\color{black}{Munich Quantum Valley (MQV)}} }\\
Email: wafa.labidi@tum.de, yaning.zhao@tu-bs.de, boche@tum.de, christian.deppe@tu-bs.de}
}

\maketitle

\sloppy

\begin{abstract} In the identification (ID) scheme proposed by Ahlswede and Dueck, the receiver's goal is simply to verify whether a specific message of interest was sent. Unlike Shannon’s transmission codes, which aim for message decoding, ID codes for a Discrete Memoryless Channel (DMC) are far more efficient: their size grows doubly exponentially with the blocklength when randomized encoding is used. This indicates that, when the receiver's objective does not require decoding, the ID paradigm is significantly more efficient than traditional Shannon transmission in terms of both energy consumption and hardware complexity.
Further benefits of ID schemes can be realized by leveraging additional resources such as feedback. In this work, we address the problem of joint ID and channel state estimation over a DMC with independent and identically distributed (i.i.d.) state sequences. State estimation functions as the sensing mechanism of the model. Specifically, the sender transmits an ID message over the DMC while simultaneously estimating the channel state through strictly causal observations of the channel output. Importantly, the random channel state is unknown to both the sender and the receiver.
For this system model, we present a complete characterization of the ID capacity-distortion function.
\end{abstract}


\section{Introduction}
The identification (ID) scheme was suggested by Ahlswede and Dueck \cite{AhlDueck} in 1989, which is conceptually different from the classical message transmission scheme proposed by Shannon \cite{Shannon}. In the classical message transmission, the encoder transmits a message over a noisy channel and at the receiver side, the aim of the decoder is to output an estimation of this message based on the channel observation. However, within the ID paradigm, the encoder sends an ID message (also called identity) over a noisy channel and the decoder aims to check whether a specific ID message of special interest to the receiver has been sent or not. Obviously, the sender has no prior knowledge of this specific ID message the receiver is interested in.  
Ahlswede and Dueck demonstrated that in the theory of ID \cite{AhlDueck}; the size of ID codes for Discrete Memoryless Channels (DMCs) grows doubly exponentially fast with the blocklength, if randomized encoding is used.  If only deterministic encoding is allowed, the number of identities that can be identified over a DMC scales exponentially with the blocklength. Nevertheless, the rate is still more significant than the transmission rate in the exponential scale as shown in \cite{deterministicDMC,IDwithoutRandom}. 

New applications in modern communication demand high reliability and latency requirements including machine-to-machine and human-to-machine systems, digital watermarking \cite{MOULINwatermarking,Ahlswede2021,SteinbergWatermarking}, industry 4.0 \cite{industry4.0} and 6G communication systems \cite{6Gcomm,6G_Book}. The aforementioned requirements are crucial for achieving trustworthiness \cite{6Gandtrustworthiness}. For this purpose, the necessary latency resilience and data security requirements must be embedded in the physical domain. In this situation, the classical Shannon message transmission is limited and an ID scheme can achieve a better scaling behavior in terms of necessary energy and needed hardware components. It has been proved that information-theoretic security can be integrated into the ID scheme without paying an extra price for secrecy \cite{AhlZhang,icassp_paper}. Further gains within the ID paradigm can be achieved by taking advantage of additional  resources such as quantum entanglement, common randomness (CR), and feedback. In contrast to the classical Shannon message transmission, feedback can increase the ID capacity of a DMC \cite{Idfeedback}. Furthermore, it has been shown in \cite{isit_paper} that the ID capacity of Gaussian channels with noiseless feedback is infinite. This holds to both rate definitions $\frac 1n \log M$ (as defined by Shannon for classical transmission) and $\frac 1n \log\log M$ (as defined by Ahlswede and Dueck for ID over DMCs). Interestingly, the authors in \cite{isit_paper} showed that the ID capacity with noiseless feedback remains infinite regardless of the scaling used for the rate, e.g., double exponential, triple exponential, etc. Besides, the resource CR allows a considerable increase in the ID capacity of channels \cite{trafo,part2,Ahlswede2021}. The aforementioned communication scenarios emphasize that the ID capacity has a completely different behavior than Shannon's capacity.

A key technology within $6$G communication systems is jointly designing radio communication and sensor technology \cite{6Gandtrustworthiness}.
This will enable the realization of revolutionary end-user applications \cite{JSC2021}. Joint communication and radar/radio sensing (JCAS) means that sensing and communication are jointly designed by sharing the same bandwidth. Sensing and communication systems are usually designed separately such that resources are dedicated to either sensing or data communications. Joint sensing and communication approach is a solution to overcome the limitations of a separation-based approach. Recent works \cite{sensing1,sensing2,sensing3,sensing4} explored JCAS and showed that this approach can improve spectrum efficiency and minimize hardware costs.
For instance, fundamental limits of joint sensing and communication for a point-to-point channel have been studied in \cite{MariPaper}, where the transmitter wishes to simultaneously send a message to the receiver and sense its channel state via a strictly causal feedback link. Motivated by the drastic effects of feedback on the ID capacity \cite{isit_paper}, this work investigates joint ID and sensing. To the best of our knowledge, the problem of joint ID and sensing has not been treated in the literature yet. We study the problem of joint ID and channel state estimation over a DMC with i.i.d. state sequences. The sender simultaneously sends an ID message over the DMC with a random state and estimates the channel state via a strictly causal channel output. The random channel state is available to neither the sender nor the receiver. We consider the ID capacity-distortion tradeoff as a performance metric. This metric is analogous to the one studied in \cite{performanceMetric} and is defined as the supremum of all ID rates we can achieve such that some distortion constraint on state sensing is fulfilled. 
The model was motivated by the problem of adaptive and sequential
optimization of the beamforming vectors during the initial access
phase of communication \cite{chiu}.
We establish a lower bound on the ID capacity-distortion tradeoff and show that, in our communication setup, sensing can be viewed as an additional resource that increases the ID capacity. 

\textit{Outline:} The remainder of the paper is organized as follows. Section~\ref{sec:Preli} introduces the system models, reviews key definitions related to identification (ID), and presents the main results, including a complete characterization of the ID capacity-distortion function. Section~\ref{sec:mainproof} provides detailed proofs of these main results. In Section~\ref{average}, we explore an alternative, more flexible distortion constraint, namely the average distortion, and establish a lower bound on the corresponding ID capacity-distortion function. Finally, Section~\ref{sec:conclusions} concludes the paper with a discussion of the results and potential directions for future research.

\textit{Notation:}
The distribution of a RV $X$ is denoted by $P_X$; for a finite set $\setx$, we denote the set of probability distributions on $\setx$ by $\mathcal{P}(\setx)$ and by $|\setx|$ the cardinality of $\setx$; if $X$ is a RV with distribution $P_X$, we denote the Shannon entropy of $X$ by $H(P_X)$, by $\mathbb{E}(X)$ the expectation of $X$ and by ${\text{Var}}[X]$ the variance of $X$; if $X$ and $Y$ are two RVs with probability distributions $P_X$ and $P_Y$, the mutual information between $X$ and $Y$ is denoted by 
 $I(X;Y)$; $\setx^c$ denotes the complement of $\setx$; $\setx-\sety$ denotes the difference set;
all logarithms and information quantities are taken to the base $2$.

\section{System Models and Main Results} \label{sec:Preli}
Let a discrete memoryless channel with random state $(\setx\times \sets, W_S(y|x,s), \sety)$ consisting of a finite input alphabet $\setx$, a finite output alphabet $\sety$, a finite state set $\sets$ and a pmf $W(y|x,s)$ on $\sety$, be given. The channel is memoryless, i.e., the probability for a sequence $y^n \in \sety^n$ to be received if the input sequence $x^n \in \setx^n$ was sent and the sequence state is $s^n \in \sets^n$ is given by 
		\begin{equation}
		W_S^n(y^n|x^n,s^n)=\prod_{i=1}^n W_S(y_i|x_i,s_i).
		\end{equation}
		 The state sequence $(S_1,S_2,\ldots,S_n)$ is i.i.d. according to the distribution $P_S$. We assume that the input $X_i$ and the state $S_i$ are statistically independent for all $i\in\{1,2,\ldots,n\}$.
	In our setting depicted in Fig. \ref{fig:System}, we assume that the channel state is  known to neither the sender nor the receiver. 
 
In the sequel, we distinguish three scenarios:
\begin{enumerate}
    \item randomized ID over the state-dependent channel $W_S$ as depicted in Fig. \ref{fig:System},
    \item deterministic or randomized ID over the state-dependent channel $W_S$ in the presence of noiseless feedback between the sender and the receiver as depicted in Fig. \ref{fig:SystemFeedback},
    \item joint deterministic or randomized ID and sensing: the sender wishes to simultaneously send an identity to the receiver and sense the channel state sequence based on the output of the noiseless feedback link as depicted in Fig. \ref{Fig:capcityDistortion}.
\end{enumerate}
\begin{figure}[hbt!]
\centering
\input{figures/System.tex}
\caption{Discrete memoryless channel with random state}
\label{fig:System}
\end{figure}
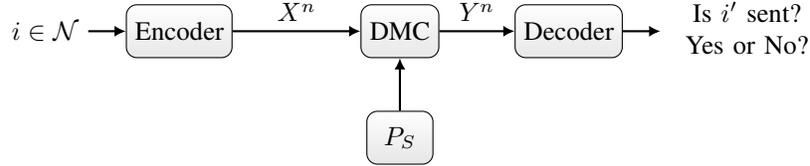
First, we define randomized ID codes for the state-dependent channel defined above.
\begin{Definition}
	An $(n,N,\lambda_1,\lambda_2)$ randomized ID code with $\lambda_1+\lambda_2<1$ for the channel $W_S$ is a family of pairs 
	$\{(Q(\cdot|i),\setd_i(s^n))_{s^n \in \sets^n}, \quad   i=1,\ldots,N\}$ with 
	\begin{equation}
	Q(\cdot|i) \in \mathcal{P}(\setx^n), \ \setd_i(s^n) \in \sety^n,\quad \forall s^n \in \sets^n, \ \forall i=1,\ldots,N,
	\end{equation}
	such that the errors of the first kind and the second kind are bounded as follows.
	\begin{align}
	\sum_{s^n \in \sets^n} P_S^n(s^n) \sum_{x^n \in \setx^n} Q(x^n|i) W_S^n(\setd_i(s^n)^c|x^n,s^n) & \leq \lambda_1,  \forall i,\\
	\sum_{s^n \in \sets^n} P_S^n(s^n) \sum_{x^n \in \setx^n} Q(x^n|i) W_S^n(\setd_j(s^n)|x^n,s^n) & \leq \lambda_2, \forall i\neq j. 
	\end{align}

\end{Definition}
In the following, we define achievable ID rate and ID capacity for our system model.
\begin{Definition}
	\begin{enumerate}
		\item The rate $R$ of a randomized $(n,N,\lambda_1,\lambda_2)$ ID code for the channel $W_S$ is $R=\frac{\log\log(N)}{n}$ bits.
		\item The ID rate $R$ for $W_S$ is said to be achievable if for $\lambda \in (0,\frac{1}{2})$ there exists an $n_0(\lambda)$, such that for all $n\geq n_0(\lambda)$ there exists an $(n,2^{2^{nR}},\lambda,\lambda)$ randomized ID code for $W_S$.
		\item The randomized ID capacity $C_{ID}(W_S)$ of the channel $W_S$ is the supremum of all achievable rates.
	\end{enumerate}
\end{Definition}
The following Theorem characterizes the randomized ID capacity of the state-dependent channel $W_S$ when the state information is known to neither the sender nor the receiver.
\begin{Theorem}
	The randomized ID capacity of the channel $W_S$ is given by
	\begin{equation}
	C_{ID}(W_S)=C(W_S)=\max_{P_X \in \mathcal{P}(\setx)} I(X;Y),
	\end{equation}
	where $C(W_S)$ denotes the Shannon transmission capacity of $W_S$. \label{proposition}
\end{Theorem}
\begin{proof}
	The proof of Theorem \ref{proposition} follows from \cite[Theorem 6.6.4]{HanBook} and \cite[eq. (7.2)]{NetworkIT}. Since the channel $W_S$ satisfies the strong converse property \cite{HanBook}, the randomized ID capacity of $W_S$ coincides with its Shannon transmission capacity determined in \cite{NetworkIT}. 
\end{proof}
Now, we consider the second scenario depicted in Fig. \ref{fig:SystemFeedback}. Let further $\bar{Y}^n=(\bar{Y}_1,\ldots,\bar{Y}_n) \in  \sety^n$ denote the output of the noiseless backward (feedback) channel. 
 \begin{equation}
    \bar{Y}_t=Y_t,\quad \forall t \in \{1,\ldots,n\}.
\end{equation}
In the following, we define a deterministic and randomized ID feedback code for the state-dependent channel $W_S$.

\begin{Definition} \label{def:IDf}
	An $(n,N,\lambda_1,\lambda_2)$ deterministic ID feedback code $\left\{(\boldsymbol{f}_i,\setd_i(s^n))_{s^n \in \sets^n},\ i=1,\ldots,N \right\}$ with $\lambda_1+\lambda_2<1$ for the channel $W_S$ is characterized as follows.
	The sender wants to send an ID message $i \in \setn := \{1,\ldots,N\} $ that is encoded by the vector-valued function
	\begin{align} 
	\boldsymbol{f}_i=[f_i^1,f_i^2\ldots,f_i^n], \label{eq: f_function}
	\end{align}
	where $f^1_i  \in \setx$ and for $t \in \{2,\ldots,n\}$, 
	$f_i^t    \colon \sety^{t-1} \longrightarrow \setx$. 
	At $t=1$ the sender sends $f_i^1$. At $t \in \{2,\ldots,n\}$, the sender sends $f^t_i({Y}_1,\ldots,{Y}^{t-1})$.
	The decoding sets $\setd_i(s^n) \subset \sety^n,\ \forall i \in \{1,\ldots,N\}, \text{and }\forall s^n \in \sets^n$ should satisfy the following inequalities.
	\begin{align}
	\sum_{s^n \in \sets^n} P_S^n(s^n) W_S^n(\setd_i(s^n)^c|\boldsymbol{f}_i,s^n) & \leq \lambda_1 \quad  \forall i, \label{eq:error1d}\\
	\sum_{s^n \in \sets^n} P_S^n(s^n) W_S^n(\setd_j(s^n)|\boldsymbol{f}_i,s^n) & \leq \lambda_2 \quad \forall i\neq j. \label{eq:error2d}
	\end{align}
	\label{Def:strategy_f}
\end{Definition}
\begin{Definition} \label{def:RIDf}
	An $(n,N,\lambda_1,\lambda_2)$ randomized ID feedback code 
    \[ \left\{(Q_F(\cdot|i),\setd_i(s^n))_{s^n \in \sets^n},\ i=1,\ldots,N \right\}\] 
    with $\lambda_1+\lambda_2<1$ for the channel $W_S$ is characterized as follows.
	The sender wants to send an ID message $i \in \setn := \{1,\ldots,N\} $ that is encoded by the probability distribution
	\begin{align} 
	Q_F(\cdot|i)\in\setp\left(\setf^n\right), \label{eq: QF_function}
	\end{align}
	where $Q_F(\cdot|i)$ denotes a probability distribution over the set of all n-length functions $\boldsymbol{f}$ as $\setf^n$. The decoding sets $\setd_i(s^n) \subset \sety^n,\ \forall i \in \{1,\ldots,N\}, \text{and }\forall s^n \in \sets^n$ should satisfy the following inequalities.
	\begin{align}
	\sum_{s^n \in \sets^n} P_S^n(s^n) \sum_{\boldsymbol{f}\in\setf^n} Q_F(\boldsymbol{f}|i) W_S^n(\setd_i(s^n)^c|\boldsymbol{f},s^n) & \leq \lambda_1 \quad  \forall i, \label{eq:error1r}\\
	\sum_{s^n \in \sets^n} P_S^n(s^n) \sum_{\boldsymbol{f}\in\setf^n} Q_F(\boldsymbol{f}|i)W_S^n(\setd_j(s^n)|\boldsymbol{f},s^n) & \leq \lambda_2 \quad \forall i\neq j. \label{eq:error2r}
	\end{align}
	\label{Def:strategy_Q}
\end{Definition}
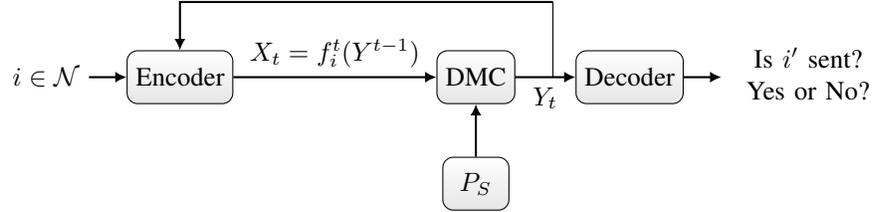
\begin{figure}[hbt!]
\centering
	\input{figures/SystemFeedback.tex}
	\caption{Discrete memoryless channel with random state and with noiseless feedback}
 \label{fig:SystemFeedback}
\end{figure}
\begin{Definition}
	\begin{enumerate}
		\item The rate $R$ of a (deterministic/randomized) $(n,N,\lambda_1,\lambda_2)$ ID feedback code for the channel $W_S$ is $R=\frac{\log\log(N)}{n}$ bits.
		\item The (deterministic/randomized) ID feedback rate $R$ for $W_S$ is said to be achievable if for $\lambda \in (0,\frac{1}{2})$ there exists an $n_0(\lambda)$, such that for all $n\geq n_0(\lambda)$ there exists a (deterministic/randomized) $(n,2^{2^{nR}},\lambda,\lambda)$ ID feedback code for $W_S$.
		\item The (deterministic/randomized) ID feedback capacity $C^d_{IDf}(W_S)$/$C^r_{IDf}(W_S)$ of the channel $W_S$ is the supremum of all achievable rates.
	\end{enumerate}
\end{Definition}
It has been demonstrated in \cite{HanBook} that noise increases the ID capacity of the DMC in the case of feedback. Intuitively, noise is considered a source of randomness, i.e., a random experiment whose outcome is provided to the sender and the receiver via the feedback channel. Thus, adding a perfect feedback link enables the realization of a  correlated random experiment between the sender and the receiver. The size of this random experiment can be used to compute the growth of the ID rate. This result has been further emphasized in \cite{isit_paper, identificationwithfeedback}, where it has been shown that the ID capacity of the Gaussian channel with noiseless feedback is infinite. This is because the authors in \cite{isit_paper,identificationwithfeedback} provided a coding scheme that generates infinite common randomness between the sender and the receiver. We want to investigate the effect of feedback on the ID capacity of our system model depicted in Fig. \ref{fig:SystemFeedback}.
Theorem \ref{MainTheorem1} characterizes the ID feedback capacity of the state-dependent channel $W_S$ with noiseless feedback. The proof of Theorem \ref{MainTheorem1} is provided in Section \ref{sec:mainproof}. 

\begin{Theorem}
	If $C(W_S)>0$, then the deterministic ID feedback capacity of $W_S$ is given by
	\begin{equation}
	C_{IDf}^d(W_S)=\max_{x\in \setx} H\left(\mathbb{E}\left[W_{S}(\cdot|x,S)\right]\right).
	\end{equation}
	\label{MainTheorem1}
\end{Theorem}
\begin{Theorem}
If $C(W_S)>0$, then the randomized ID feedback capacity of $W_S$ is given by
	\begin{equation}
	C_{IDf}^r(W_S)=\max_{P\in \setp\left(\setx\right)} H\bigg(\sum_{x\in\setx}P(x)\mathbb{E}\left[W_S(\cdot|x,S)\right]\bigg).
	\end{equation}
	\label{MainTheorem2}
\end{Theorem}
\begin{Remark}
It can be shown that the same ID feedback capacity formula holds if the channel state is known to either the sender or the receiver. This is because we achieve the same amount of common randomness as in the scenario depicted in Fig. \ref{fig:SystemFeedback}. Intuitively, the channel state is an additional source of randomness that we take advantage of.
\end{Remark}

Now, we consider the third scenario depicted in Fig. \ref{Fig:capcityDistortion}, where we want to jointly identify and sense the channel state. The sender comprises an encoder that sends a symbol $x_t=f_i^t(y^{t-1})$ for each identity $i \in \{1,\ldots,N\}$ and delayed feedback output $y^{t-1} \in \sety^{t-1}$ and a state estimator that outputs an estimation sequence $\hat{s}^n \in {\sets}^n$ based on the feedback output and the input sequence. We define the per symbol distortion as follows:
\begin{equation}
 d_t=\mathbb{E}\left[d(S_t,\hat{S}_t)\right],\label{eq:persymbolDistortion}
\end{equation}
where $d\colon \sets \times {\sets} \to [0, +\infty)$ is a distortion function and the expectation is over the joint distribution of $(S_{t},\hat{S}_{t})$ conditioned by the ID message $i \in \{1,\ldots,N\}$.
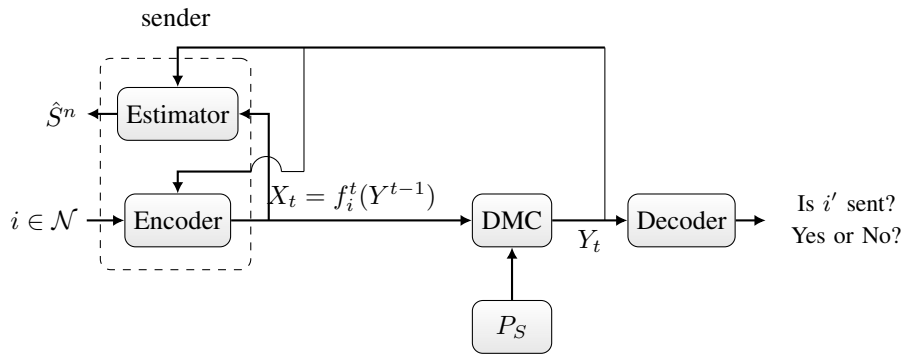
\begin{figure}[hbt!]
\centering
\input{figures/JointSensing.tex}
\caption{State-dependent channel with noiseless feedback}
\label{Fig:capcityDistortion}
\end{figure}
\begin{Definition}
    \begin{enumerate} 
        \item An ID rate-distortion pair $(R,D)$ for $W_S$ is said to be achievable if for every $\lambda \in (0,\frac{1}{2})$ there exists an $n_0(\lambda)$, such that for all $n \geq n_0(\lambda)$ there exists an $(n,2^{2{nR}}, \lambda,\lambda)$ (deterministic/randomized) ID code for $W_S$ and for all $t=1,\cdots,n$, $d_t \leq D$. 
        \item The deterministic ID capacity-distortion function $C_{ID}^d(D)$ is defined as the supremum of $R$ such that $(R,D)$ is achievable.
    \end{enumerate}
\end{Definition}
We choose without loss of generality the following deterministic estimation function $h^\star$
\begin{equation}
    \hat{s}=h^\star(x,y)=\min_{h:\setx\times\sety\to \sets}\mathbb{E}\left[d(S,h(X,Y))|X=x,Y=y\right],
\end{equation}
where $h\colon \setx\times \sety \to \sets$ is an estimator that maps a pair of channel input and feedback output to a channel state. If there exist several functions $h^\star(\cdot,\cdot)$, we choose one randomly. 
We define the minimal distortion function for each input symbol $x \in \setx$ as in \cite{oneShotEstimator}:
\begin{equation}
    d^\star(x)=\mathbb{E}_{SY}[d(S,h^*(X,Y))|X=x], \label{eq:minimalDis1}
\end{equation}
and the minimal distortion function for each input distribution $P\in \setp\left(\setx\right)$:
\begin{equation}
    d^\star(P)=\sum_{x\in\setx}d^\star(x).\label{eq:minimalDis2}
\end{equation}
In the following, we establish the ID capacity-distortion function defined above.
\begin{Theorem}
The deterministic ID capacity-distortion function of the state-dependent channel $W_S$ depicted in Fig. \ref{Fig:capcityDistortion} is given by:
\begin{equation}
    C_{ID}^d (D) = \max_{x \in \setx_D}  H\big(\mathbb{E}[W_S(\cdot|x,S)] \big),
\end{equation}
where the set $\setx_D$ is given by
\begin{equation}
    \setx_D=\{x \in \setx,\quad d^\star(x) \leq D\}.
\end{equation}
\label{theorem:capDist1}
\end{Theorem}
We now turn our attention to a randomized encoder. In the following, we derive the ID capacity-distortion function of the state-dependent channel $W_S$ under the assumption of randomized encoding.
\begin{Theorem}
The randomized ID capacity-distortion function of the state-dependent channel $W_S$ is given by:
\begin{equation}
    C_{ID}^r (D) = \max_{P \in \setp_D}  H\big(\sum_{x\in\setx}P(x)\mathbb{E}[W_S(\cdot|x,S)] \big),
\end{equation}
where the set $\setp_D$ is given by
\begin{equation}
    \setp_D=\{P \in \setp(\setx),\quad d^\star(P) \leq D\}.
\end{equation}
\label{theorem:capDist2}
\end{Theorem}
\begin{Remark}
   Randomized encoding achieves higher rates than deterministic encoding. This is because we are combining two sources of randomness; local randomness used for the encoding and the shared randomness generated via the noiseless feedback link. The result is analogous to randomized ID over DMCs in the presence of noiseless feedback studied in \cite{Idfeedback}.

\end{Remark}
\section{Proof of the Main Results} \label{sec:mainproof}
In this section, we provide the proofs of Theorem \ref{MainTheorem1}, Theorem \ref{MainTheorem2}, Theorem \ref{theorem:capDist1}, and Theorem \ref{theorem:capDist2}.


\subsection{Direct Proof of Theorem \ref{MainTheorem1}}
We consider an average channel $W_{S,\text{avg}}$ given by
\begin{equation}
W_{S,\text{avg}}(y|x)= \sum_{s \in \sets} P_S(s) W_S(y|x,s), \quad \forall x \in \setx,\ \forall y \in \sety. \label{eq:averageChannel}
\end{equation}
The DMC $W_{S,\text{avg}}$ is obtained by averaging the DMCs $W_S$ over the state. Now, it suffices to show that $R=\max_{x \in \setx} H\big(\mathbb{E}[W_S(\cdot|x,S)] \big)$ is an achievable ID feedback rate for the average channel $W_S^\text{a}$.
 The deterministic ID feedback capacity of the average channel $C_{IDf}^d(W_{S,\text{avg}})$ can be determined by applying \cite[Theorem 1]{identificationwithfeedback} on $W_{S,\text{avg}}$. If the transmission capacity $C(W_{S,\text{avg}})$ of $W_{S,\text{avg}}$ is positive, we have
\begin{align}
C_{IDf}^d(W_{S,\text{avg}}) & \leq \max_{x \in \setx} H\big(W_{S,\text{avg}}(\cdot|x)\big) \\
& = \max_{x \in \setx} H\big(\mathbb{E}[W_S(\cdot|x,S)] \big).
\end{align}
This completes the direct proof of Theorem \ref{MainTheorem1}. \qed

\subsection{Converse Proof of Theorem \ref{MainTheorem1}}
\label{sec_converse_IDS}
For the converse proof, we use the techniques of \cite{Idfeedback} for deterministic ID over DMCs with noiseless feedback. We first extend \cite[Lemma~3]{Idfeedback} (image size for a deterministic feedback strategy) to the deterministic ID feedback code for $W_S$ described in Definition \ref{def:IDf}. 
\begin{Lemma}
For any feedback strategy $\boldsymbol{f}=[f^1,f^2\ldots,f^n]$ and any $\mu \in (0,1)$, we have
\begin{equation}
    \min_{\sete_1\subset \sety^n\colon \mathbb{E}_{S^n}\big[W^n_S(\sete_1|\boldsymbol{f},S^n)\big] \geq 1-\mu} |\sete_1| \leq K_1, \label{eq:cardinalityOfE}
\end{equation}
where $K_1$ is given by
\begin{equation}
    K_1
    =2^{n\max_{x\in\setx}H\left(\mathbb{E}\left[W_S(\cdot|x)\right]\right)+\alpha\sqrt{n}}
    =2^{nH\left(\mathbb{E}\left[W_S(\cdot|x^{\star})\right]\right)+\alpha\sqrt{n}}, \label{eq:definitionOfK}
\end{equation}
where  \begin{equation}
    \alpha=\sqrt{\frac{\beta}{\mu}}, \quad \beta=\max(\log^2(3),\log^2(|\sety|) \label{eq:DefinitionofBeta}.
\end{equation}
\label{lemma3}
\end{Lemma}
\begin{proof}
We use a similar idea as for the proof of \cite[Lemma 3]{Idfeedback}. Let $\sete_1^\star \subset \sety^n$ be defined as follows:
\begin{equation}
    \sete_1^{\star}=\bigg \{ y^n \in \sety^n, \quad -\log\mathbb{E}_{S^n}\big[W^n_S(y^n|\boldsymbol{f},S^n)\big] \leq \log K_{1} \bigg\}. \label{eq:definitionOfE}
\end{equation}
It follows from the definition of $\sete_{1}^\star$ that $|\sete^\star|\leq K_1$. It remains to show that 
\[\mathbb{E}_{S^n}\big[W_S^n(y^n|\boldsymbol{f},S^n)\big] \geq 1-\mu.\] For $t=1,2,\ldots,n$ and a fixed feedback strategy $\boldsymbol{f}=[f^1,\ldots,f^t] \in \setx$, we have
\begin{align}
    \Pr\{Y^t=y^t\} & =\sum_{s^n \in \sets^n} P_S^n(s^n)W^n(y^n|\boldsymbol{f},s^n) \\
    & = \mathbb{E}_{S^n}\big[ W_S^n(y^n|\boldsymbol{f},S^n) \big], \quad \forall y^t \in \sety^t.
\end{align}
Let the RV $Z_t$ be defined as follows:
\begin{equation}
    Z_t=-\log \mathbb{E}_{S_t}\big[ W_S(Y_t|f^t(Y^{(t-1)},S_t) \big].
\end{equation}
We have
\begin{align}
 \mathbb{E}_{S^n}\big[ W_S^n(\sete_1^\star|\boldsymbol{f},S^n) \big]
 & = \Pr\bigg\{  -\log \mathbb{E}_{S^n}\big[ W_S^n(y^n|\boldsymbol{f},S^n) \big] \leq \log K_1  \bigg\} \nonumber\\
 &=  \Pr\bigg\{ -\log \big(\sum_{s^n \in \sets^n} P_S^n(s^n) \prod_{t=1}^n W_S(y_t|f^t,s_t) \big) \leq \log K_1   \bigg\} \nonumber\\
 &= \Pr\bigg\{ -\log \big( \sum_{s_1 \in \sets} \sum_{s_2 \in \sets} \cdots \sum_{s_n \in \sets} \prod_{t=1}^n P_S(s_t) W_S(y_t|f^t,s_t)  \big) 
  \leq \log K_1 \bigg \}\nonumber \\
 &= \Pr \bigg \{  -\log \big( \sum_{s_1 \in \sets} P_S(s_1) W_S(y_1|f^1,s_1)  \nonumber\\
 & \quad \cdot \sum_{s_2 \in \sets} P_S(s_2) W_S(y_2|f^2,s_2)\cdots \sum_{s_n \in \sets} P_S(s_n) W_S(y_n|f^n,s_n)   \big)  \leq \log K_1 \bigg\} \nonumber\\
 &=\Pr \bigg \{  \sum_{t=1}^n -\log\mathbb{E}_{S_t}\big[ W_S(Y_t|f^t(Y^{(t-1)},S_t) \big] \leq \log K_1\bigg \}\nonumber\\
 &= \Pr\bigg \{ \sum_{t=1}^n Z_t \leq \log K_1 \bigg\}.\label{eq:Z_t}
\end{align}
Now, we want to establish a lower bound on $\mathbb{E}_{S^n}\big[ W_S^n(\sete^\star|\boldsymbol{f},S^n) \big]$. It suffices to find a lower bound on the expression in \eqref{eq:Z_t}. Let the RV $U_t$ be defined as follows:
\begin{equation}
U_t=Z_t-\mathbb{E}[Z_t|Y^{t-1}], \quad t\in \{1,2,\ldots,n\}. \label{eq:definitionOfUt}
\end{equation}
It can be shown that
\begin{align}
\mathbb{E}[U_t|Y^{t-1}] &= \mathbb{E}_{Y}\bigg[ Z_t-\mathbb{E}_Y[Z_t|Y^{t-1}] | Y^{t-1} \bigg] \nonumber\\
&= \mathbb{E}_{Y}[ Z_t|Y^{t-1} ]- \mathbb{E}_Y[Z_t|Y^{t-1}] \nonumber\\
&=0.
\label{eq:u_t1}
\end{align}
Furthermore, we have
\begin{align}
\mathbb{E}[U_t]&=\mathbb{E}_{Y}\bigg[ Z_t-\mathbb{E}_Y[Z_t|Y^{t-1}] \bigg] \nonumber\\
&= \mathbb{E}_{Y}[ Z_t]- \mathbb{E}_Y\bigg[\mathbb{E}_Y[Z_t|Y^{t-1}]\bigg] \nonumber\\
&=\mathbb{E}_{Y}[ Z_t]- \mathbb{E}_{Y}[ Z_t] \nonumber\\
&=0.
\label{eq:u_t2}
\end{align}
Furthermore, it can be shown that for $y^{t-1} \in \sety^{t-1}$
\begin{align}
\mathbb{E}[Z_t|y^{t-1}] 
&  = \sum_{y_t \in \sety} \bigg( -\mathbb{E}_{S_t}\big[W_S(y_t|f^t(y^{t-1}),S_t)\big] \nonumber \\
& \quad \log \mathbb{E}_{S_t}\big[W_S(y_t|f^t(y^{t-1}),S_t)\big] \bigg) \nonumber \\
&\leq \max_{x\in \setx}H\bigg(\mathbb{E}_{S}\big[W_S(\cdot|x,S)\big] \bigg) \nonumber\\
&=H\bigg(\mathbb{E}_{S}\big[W_S(\cdot|x^\star,S)\big] \bigg). \label{eq:BoundonEZ}
\end{align}
It follows from the definition of the RV $U_t$ in \eqref{eq:definitionOfUt} that
\begin{align}
    \mathbb{E}_{S^n}\big[ W_S^n(\sete^\star|\boldsymbol{f},S^n) \big] 
    & = \Pr\bigg \{ \sum_{t=1}^n \left(U_t+ \mathbb{E}[Z_t|Y^{t-1}]\right)\leq \log K_1 \bigg\} \nonumber\\
    &= \Pr\bigg \{ \sum_{t=1}^n U_t\leq \log K_1 - \sum_{t=1}^{n}\mathbb{E}[Z_t|Y^{t-1}] \bigg\} \nonumber\\
    & \overset{(a)}{=}  \Pr\bigg \{ \sum_{t=1}^n U_t\leq nH\big(\mathbb{E}_{S}\big[W_S(\cdot|x^\star,S)\big] \big) \nonumber \\
    & \quad +\alpha\sqrt{n} - \sum_{t=1}^{n}\mathbb{E}[Z_t|Y^{t-1}] \bigg\} \nonumber\\
    &\overset{(b)}{\geq} \Pr\bigg \{ \sum_{t=1}^n U_t\leq \alpha\sqrt{n} \bigg\},
\end{align}
where $(a)$ follows from the definition of $K_1$ in \eqref{eq:definitionOfK} and $(b)$ follows from \eqref{eq:BoundonEZ}. \\
It can be verified that 
\begin{equation}
    \text{Var}[U_t] \leq \beta, \quad \forall t=1,2,\ldots,n.
\end{equation}
Therefore, we can apply Chebyshev's inequality and obtain
\begin{align}
    \mathbb{E}_{S^n}\big[ W_S^n(\sete^\star|\boldsymbol{f},S^n) \big] & \geq  \Pr\bigg \{ \sum_{t=1}^n U_t\leq \alpha\sqrt{n} \bigg\} \nonumber\\
    & \overset{(a)}{\geq} 1-\mu, \nonumber
\end{align}
where $(a)$ follows the definition of $\beta$ in \eqref{eq:DefinitionofBeta}. This completes the proof of Lemma \ref{lemma3}.
\end{proof}

We establish an upper bound on the deterministic ID feedback rate for the channel model $W_S$ using Lemma \ref{lemma3}.
Let $\left\{(\boldsymbol{f}_i,\setd_i(s^n))_{s^n \in \sets^n},\quad i=1,\ldots,N \right\}$ be an $(n,N,\lambda,\lambda)$ deterministic ID feedback code for the channel $W_S$ with $\lambda \in (0,\frac{1}{2})$ and let $\mu \in (0,1)$ be chosen such that 
\begin{equation}
    1-\mu-\lambda<\frac{1}{2}. \label{eq:DefofMu}
\end{equation}
We define for each feedback strategy $\boldsymbol{f}_i$ the set $\sete_{i}$ that satisfies \eqref{eq:cardinalityOfE}. For $ i\in \{1,2,\ldots,N\}$, we have
\begin{align}
\mathbb{E}_{S^n}\big[ W_S^n(\setd_i(s^n)\cap\sete_i|\boldsymbol{f}_i,S^n) \big] &
=  1 - \mathbb{E}_{S^n}\big[ W_S^n( (\setd_i(s^n)\cap\sete_i)^c|\boldsymbol{f}_i,S^n) \big] \\
&=1- \mathbb{E}_{S^n}\big[ W_S^n((\setd_i(s^n))^c\cup \sete_i^c|\boldsymbol{f}_i,S^n) \big] \\
& \overset{(a)}{\geq} 1 - \mathbb{E}_{S^n}\big[ W_S^n((\setd_i(s^n))^c|\boldsymbol{f}_i,S^n) \big] - \mathbb{E}_{S^n}\big[ W_S^n(\sete_i^c|\boldsymbol{f}_i,S^n) \big] \\
& \overset{(b)}{\geq} 1-\lambda-\mu \\
&\overset{(c)}{>} \frac{1}{2},
\end{align}
where $(a)$ follows from the Union Bound, $(b)$ follows from the definition of the $(n,N,\lambda,\lambda)$ ID feedback code and from \eqref{eq:cardinalityOfE} and $(c)$ follows from \eqref{eq:DefofMu}.
Similarly, it follows from the definition of the ID feedback code 
\[ 
\left\{(\boldsymbol{f}_i,\setd_i(s^n))_{s^n \in \sets^n},\quad i=1,\ldots,N \right\},\] 
from \eqref{eq:cardinalityOfE} and \eqref{eq:DefofMu} that for $i \in \{1,2,\ldots,N\}$ and $i \neq j$
\begin{equation}
    \mathbb{E}_{S^n}\big[ W_S^n(\setd_j(s^n)\cap\sete_{j}|\boldsymbol{f}_i,S^n) \big] < \frac{1}{2}.
\end{equation}
As the error of the second kind $\lambda$ is smaller than $\frac{1}{2}$, all the sets $\setd_i(s^n)\cap\sete_{i}$ are distinct. Therefore, any $(n,N,\lambda,\lambda)$ deterministic ID feedback code $\left\{(\boldsymbol{f}_i,\setd_i(s^n))_{s^n \in \sets^n},\quad i=1,\ldots,N \right\}$ for the channel $W_S$ with $\lambda \in (0,\frac{1}{2})$ has an associated $(n,N,\lambda^\prime,\lambda^\prime)$ deterministic ID feedback code $\left\{(\boldsymbol{f}_i,\setd_i(s^n))_{s^n \in \sets^n} \cap \sete_{i},\quad i=1,\ldots,N \right\}$, where $\lambda^\prime \in (0,\frac{1}{2})$ and $\forall i=1,\ldots,N$ the set $\sete_{i}$ satisfies \eqref{eq:cardinalityOfE}. Thus, by Lemma \ref{lemma3}, the cardinality $N$ of the deterministic ID feedback code is upper-bounded as follows:
\begin{align}
N & \leq \sum_{k=0}^{K_1} {{|\sety|^n}\choose{k}} \leq \big (|\sety|^n \big )^{K_1} \\
& = 2^{n \log |\sety| 2^{n  H\big(\mathbb{E}[W_S(\cdot|x^\star,S)] \big) + \alpha \sqrt{n}}}.
\end{align}
This completes the converse proof of Theorem \ref{MainTheorem1}.\qed


\subsection{Direct Proof of Theorem \ref{MainTheorem2}}
Similarly, we consider the average channel $W_{S,avg}$ defined in \eqref{eq:averageChannel}. It is sufficient to show that $R=\max_{P\in\setp(\setx)}H(\sum_{x\in\setx}P(x)W_{S,avg}(\cdot|x))$ is an achievable randomized ID feedback rate for the average channel $W_{S,avg}$. The randomized ID feedback capacity of the average channel $C^{r}_{IDf}(W_{S,avg})$ can be obtained by applying \cite[Theorem 2]{identificationwithfeedback} on $W_{S,avg}$. If the transmission capacity $C(W_{S,avg})$ of $W_{S,avg}$ is positive, then we have
\begin{align}
    C^{r}_{IDf}(W_{S,avg})
    &\leq \max_{P\in\setp(\setx)}H\left(\sum_{x\in\setx}P(x)W_{S,avg}(\cdot|x)\right).
\end{align}
This completes the direct proof of Theorem \ref{MainTheorem2}.\qed

\subsection{Converse Proof of Theorem \ref{MainTheorem2}}

We first extend \cite[Lemma 4]{Idfeedback} (image size for a randomized feedback strategy) to the randomized ID feedback code for $W_S$.

\begin{Lemma}
    For any randomized feedback strategy $Q_F(\cdot)$ over all n-length feedback encoding set $\setf_n$ and any $\mu\in(0,1)$, 
    \begin{align}
        \min_{\sete_2\subset\sety^n:\sum_{\boldsymbol{f}\in\setf_n}Q_F(\boldsymbol{f})\mathbb{E}_{S^n}\left[W^n_S(\sete_2|\boldsymbol{f},S^n)\right]\geq 1-\mu}|\sete_2|\leq K_2,
    \end{align}
    where $K_{2}$ is given by 
    \begin{align}
        K_{2}=2^{n\max_{P\in\setp\left(\setx\right)}H\left(\sum_{x\in\setx}P(x)\mathbb{E}\left[W_S\left(\cdot|x,S\right)\right]\right)+\alpha\sqrt{n}}=2^{nH\left(\sum_{x\in\setx}P^{\star}(x)\mathbb{E}\left[W_S(\cdot|x,S)\right]\right)+\alpha\sqrt{n}},
    \end{align}
    where $\alpha=\sqrt{\frac{\beta}{\mu}}$, $\beta=\max{\log^2{3},\log^2{|\sety|}}$.
    \label{lemma:K2}
\end{Lemma}
\begin{proof}
    We use a similar idea as for the proof of \cite[Lemma 4]{Idfeedback}. Define the set $\sete_{2}^{\star}\subset\sety^n$ as follows:
    \begin{align}
        \sete_{2}^{\star}=\left\{y^n\in\sety^{n}, -\log{\sum_{\boldsymbol{f}\in\setf_n}Q_F\left(\boldsymbol{f}\right)\mathbb{E}_{S^n}\left[W_S(y^n|\boldsymbol{f},S^n)\right]}\leq \log{K_2}\right\}.
    \end{align}

By the definition of $\sete_{2}^{\star}$ we have 
\[ \left|\sete_{2}^{\star}\right|\leq K_2.\] 
Then it is sufficient to show that $\sum_{\boldsymbol{f}\in\setf_n}Q_F(\boldsymbol{f})\mathbb{E}_{S^n}\left[W^n_S(\sete_2|\boldsymbol{f},S^n)\right]\geq 1-\mu$.
Let the RV $Z_t$ be defined as follows:
 \begin{equation}
     Z_t=-\log\left(\sum_{f^t\in\setf^t}Q^{t}(f^t)\mathbb{E}_{S_t}\left[W^n_S(Y_t|f^t(Y^{t-1}),S_t)\right]\right),
 \end{equation}
 where $Q_{F}(\boldsymbol{f})=\prod_{t=1}^n Q^{t}(f^t)$ and $\setf^t$ is the set of all mapping $\sety^{t-1}\mapsto\setx$.
We have
\begin{align}
    &\sum_{\boldsymbol{f}\in\setf_n}Q_F(\boldsymbol{f})\mathbb{E}_{S^n}\left[W_S^n\left(\sete_2|\boldsymbol{f},S^n\right)\right]\\
    &=Pr\left\{-\log\left({\sum_{\boldsymbol{f}\in\setf_n}Q_F\left(\boldsymbol{f}\right)\mathbb{E}_{S^n}\left[W_S(y^n|\boldsymbol{f},S^n)\right]}\right)\leq \log{K_2}\right\}\nonumber\\
    &=Pr\left\{-\log\left(\sum_{\boldsymbol{f}\in\setf_n}Q_F\left(\boldsymbol{f}\right)\prod_{t=1}^{n}\mathbb{E}_{S_t}\left[W_S(Y_t|f^t\left(Y^{t-1}\right),S_t\right]\right)\leq \log{K_2}\right\}\nonumber\\
    &=Pr\left\{-\log\left(\sum_{\boldsymbol{f}\in\setf_n}\prod_{t=1}^n Q^{t}\left(f^t\right)\prod_{t=1}^{n}\mathbb{E}_{S_t}\left[W_S(Y_t|f^t\left(Y^{t-1}\right),S_t\right]\right)\leq \log{K_2}\right\}\nonumber\\
    &=Pr\left\{-\log\left(\prod_{t=1}^n \sum_{f^t\in\setf^t}Q^{t}\left(f^t\right)\mathbb{E}_{S_t}\left[W_S(Y_t|f^t\left(Y^{t-1}\right),S_t\right]\right)\leq \log{K_2}\right\}\nonumber\\
    &=Pr\left\{\sum_{t=1}^{n}Z_t\leq \log{K_2}\right\}.
    \label{eq:rid1}
\end{align}
Now for any $y^{t-1}\in\sety^{t-1}$, we consider
\begin{align}
    &\mathbb{E}\left[Z_t|y^{t-1}\right]\\
    &=\sum_{y_t\in\sety}(-\sum_{f^t\in\setf^t}Q^{t}(f^t)\mathbb{E}_{S_t}[W^n_S(y_t|f^t(y^{t-1}),S_t)]\log(\sum_{f^t\in\setf^t}Q^{t}(f^t)\mathbb{E}_{S_t}[W^n_S(y_t|f^t(y^{t-1}),S_t)]))\nonumber\\
    &\leq H\left(\sum_{x\in\setx}P^{\star}(x)\mathbb{E}\left[W_S^n(\cdot|x,S\right]\right).
    \label{eq:rid2}
\end{align}
Similarly, for all $t\in\left\{1,2,\cdots,n\right\}$, we define RV $U_t=Z_t-\mathbb{E}\left[Z_t|Y^{t-1}\right]$. It has been shown in \eqref{eq:u_t1} and \eqref{eq:u_t2} that $\mathbb{E}\left[U_t|Y^{t-1}\right]=0$ and $\mathbb{E}\left[U_t\right]=0$. 
Combine \eqref{eq:rid1} and \eqref{eq:rid2}, we have
\begin{align}
    &\sum_{\boldsymbol{f}\in\setf_n}Q_F(\boldsymbol{f})\mathbb{E}_{S^n}\left[W_S^n\left(\sete_2|\boldsymbol{f},S^n\right)\right]\\
    &=Pr\left\{\sum_{t=1}^{n}\left(U_t+\mathbb{E}\left[Z_t|Y^{t-1}\right]\right)\leq \log{K_2}\right\}\nonumber\\
    &=Pr\left\{\sum_{t=1}^{n}U_t\leq nH\left(\sum_{x\in\setx}P^{\star}(x)\mathbb{E}\left[W_S(\cdot|x,S)\right]\right)+\alpha\sqrt{n}-\sum_{t=1}^n\mathbb{E}\left[Z_t|Y^{t-1}\right]\right\}\nonumber\\
    &\geq Pr\left\{\sum_{t=1}^{n}U_t\leq \alpha\sqrt{n}\right\}.
\end{align}
Assume for all $t=1,2,\cdots,n$, $Var[U_t]\le \beta$. By applying Chebyshev's inequality, we can obtain
\begin{align}
    \sum_{\boldsymbol{f}\in\setf_n}Q_F(\boldsymbol{f})\mathbb{E}_{S^n}\left[W_S^n\left(\sete_2|\boldsymbol{f},S^n\right)\right]
    &\geq 1-\frac{\beta}{\alpha^2}\nonumber\\
    &=1-\mu.
\end{align}
Replace $K_1$ in the converse proof of Theorem \ref{MainTheorem1} with the corresponding $K_2$ as outlined in Lemma \ref{lemma:K2}, completing the converse proof of Theorem \ref{MainTheorem2}.
\end{proof}

\subsection{Direct Proof of Theorem \ref{theorem:capDist1}}
\label{sec: direct JdIDAS}
\subsubsection{Coding Scheme}
We use some extent the same coding scheme elaborated in \cite{Idfeedback}.
 We choose as blocklength $m=n+\lceil \sqrt{n} \rceil$. Let $x^\star$ be some symbol in $\setx_D$. Regardless of which identity $i \in \{1,\ldots,N\}$ we want to identify, the sender first sends the sequence ${x^\star}^n=(x^\star,x^\star,\ldots,x^\star) \in \setx_D^n$ over the state-dependent channel $W_S$. The received sequence $y^n \in \sety^n$ becomes known to the sender (estimator and encoder) via the noiseless feedback link. The feedback provides the sender and the receiver with the knowledge of the outcome of the correlated random experiment $\bigg( \sety^n, \mathbb{E}_{S^n}\big[W_S^n(\cdot|{x^\star}^n,S^n) \big] \bigg)$. 
 \subsubsection{Common Randomness Generation}
 We want to generate uniform common randomness because it is the most convenient form of common randomness \cite{part2}. Therefore, we convert our correlated random experiment $\bigg( \sety^n, \mathbb{E}_{S^n}\big[W_S^n(\cdot|{x^\star}^n,S^n) \big] \bigg)$ to a uniform one $\bigg( \mathcal{T}^n, \mathbb{E}_{S^n}\big[W_S^n(\cdot|{x^\star}^n,S^n) \big] \bigg)$. For $\epsilon >0$, the set $\mathcal{T}^n$ is given by
 \begin{equation}
     \mathcal{T}^n= \bigcup_{V_S^a \colon \norm{V_S^a-W_S^a}\leq \epsilon} \mathcal{T}_{V_S^a}^n ({x^\star}^n), \label{eq:definitionOfTn_d}\end{equation}
where $\norm{V_S^a-W_S^a}$ is defined as follows:
\begin{Definition}Let $\mathcal{W}$ be the set of stochastic matrices $W\colon \setx \rightarrow \sety$. Let $W \in \mathcal{W}$ such that 
\begin{equation}
    P_{XY}(x,y)=P_X(x) W(y|x), \quad \forall x \in \setx, \forall y \in \sety.
\end{equation}
For $V,V^\prime \in \mathcal{W}$, the distance $|V-V^\prime|$ is defined as follows.
\begin{equation}
    \norm{V-V^\prime} = \max_{x \in \setx, y\in \sety} |V(y|x)-V^\prime(y|x)|.
\end{equation}
\end{Definition}
$V_S^a$ denotes the average channel defined by
\begin{equation}
    V_S^a=\sum_{s \in \sets} P_S(s) V_S(\cdot|\cdot,s).
\end{equation}
We introduce the following lemmas.
\begin{Lemma}{\cite{IT_CiKo}}
  Let $(x^n,y^n)$ be emitted by the DMS $P_{XY}(\cdot,\cdot)= W(\cdot|\cdot)P_X(\cdot)$ and let $V \in \mathcal{W}$ such that $|V-W|\leq \epsilon$. Then for every $\epsilon >0$, there exist $\delta^\prime>0$ and $n_0(\epsilon)$ such that for $n \geq n_0(\epsilon)$
    \begin{equation}
    \sum_{y^n \in \sett^n_V(x^n)} W^n(y^n|x^n) \geq 1-2^{-n \delta^\prime}.
    \end{equation} 
    \label{lemma:typ1}
\end{Lemma}
\begin{Lemma}{\cite{IT_CiKo}}
 Let $(x^n,y^n)$ be emitted by the DMS $P_{XY}(\cdot,\cdot)= W(\cdot|\cdot)P_X(\cdot)$ and let $V \in \mathcal{W}$. For every $\epsilon >0$ there exist a $c(\epsilon)>0$ and $n_0(\epsilon)$ such that for $n\geq n_0(\epsilon)$
\begin{enumerate}
\item $ | \underset{V \colon \norm{V-W} \leq \epsilon}{\bigcup} \sett^n_V(x^n) | \geq 2^{n (H(W|P_X) - c(\epsilon) )} $
\item $ | \underset{V \colon \norm{V-W} \leq \epsilon}{\bigcup} \sett^n_V(x^n) | \leq 2^{n (H(W|P_X) +c(\epsilon) )} $
\item if $\norm{V-W}\leq \epsilon$, $\sett_V^n(x^n) \neq \emptyset$ and $c(\epsilon) \to 0$ if $\epsilon \to 0$, then
\begin{equation}
    |\sett_V^n(x^n)| \geq 2^{n (H(W|P_X) - c(\epsilon) )}.
\end{equation}
\end{enumerate}
\label{lemma:typ2}
\end{Lemma}

\begin{Lemma}\cite{Hoeffding}  \label{HoeffIneq}
 Let $\{X_i\}$ be i.i.d. RV's taking values in $[0,1]$ with mean $\mu$, then $\forall c >0$ with $p=\mu+c \leq 1$, 
 \begin{equation}
 \Pr\{\bar{X}_n - \mu \geq c \} \leq \exp({-n D(p||\mu)}) \leq \exp({-2nc^2}),
 \end{equation}
 where $\bar{X}_n = \frac{1}{n} \sum_{i=1}^{n} X_i$.
\end{Lemma}
\label{app:proofConcavity}
For arbitrary $D_{min} \leq D_1 \leq D_2$, $\setx_{D_1}$ and $\setx_{D_2}$ are defined by
\begin{align}
   \setx_{D_1} &=\{x \in \setx,\quad d^\star(x) \leq D_1\} \\
   \setx_{D_2} & =\{x \in \setx,\quad d^\star(x) \leq D_2\}.
\end{align}
It is clear that $g(D)$ is a non-decreasing function because $\setx_{D_1} \subseteq \setx_{D_2}$ for arbitrary $D_1 \leq D_2$. Let $\mu \in (0,1)$, we have
\begin{align}
   g\big(\mu D_1 + (1-\mu)D_2\big) &= \max_{x \in \setx_{\mu D_1 + (1-\mu)D_2}}  H\big(\mathbb{E}[W_S(\cdot|x,S)] \big) 
\end{align}

It follows from Lemma \ref{lemma:typ2} that $\mathbb{E}_{S^n}\big[W_S^n(\cdot|{x^\star}^n,S^n) \big] $ is essentially uniform on $\mathcal{T}^n$. Let the set $\sety^n -\sett^n$ be denoted by $\sete^\star$. By Lemma \ref{lemma:typ1}, we have
\begin{align}
   \Pr\{ Y^n \in \sete^\star|X^n=({x^\star}^n,{x^\star}^n,\ldots,{x^\star}^n) \} 
   & = 1-  \Pr\{ Y^n \notin \sett^n|X^n=({x^\star}^n,{x^\star}^n,\ldots,{x^\star}^n) \}\nonumber\\
    & \leq 2^{-n \delta^\prime}. \label{eq:errorPb}
\end{align}
  As we mentioned earlier, we have $|\sett^n|\approx 2^{n \mathbb{E}_S\big[W_S(\cdot|{x^\star}^n,S^n) \big]}$. This quantity is the size of the correlated random experiment $(\sett^n,\mathbb{E}_{S^n}\big[W_S^n(\cdot|{x^\star}^n,S^n) \big])$, which determines the growth of the ID rate. Let $\mathcal{C}=\{({\boldsymbol{u}}_j,\setd_j),\ j=1,\ldots,M\}$ be an $(\lceil \sqrt{n} \rceil,M,2^{-\sqrt{n} \delta})$ code, where ${\boldsymbol{u}}_j \in \setx_D^{\lceil \sqrt{n} \rceil}$ for each $j=1,\ldots,M$. We concatenate the sequence ${x^\star}^n=(x^\star,x^\star,\ldots,x^\star) \in \setx_D^n$ and the transmission code $\mathcal{C}$ to build an $(m,N,\lambda_1,\lambda_2)$ ID feedback code $\Cp=\{(\boldsymbol{f}_i,\setd'_i), \ i=1,\ldots,N\}$ for $W_S$. We have $\lambda_1,\lambda_2<\lambda, \quad \lambda \in (0,\frac{1}{2})$. The concatenation is done using the  coloring functions, $\{F_i,\ i=1,\ldots,N \}$. We choose a suitable set of coloring functions $\{F_i, \quad i=1,\dots,N\}$ randomly. Every coloring function $F_i\colon  \sett^n \longrightarrow \{1,\ldots,M \} $ corresponds to an ID message $i$ and maps each element $y^n \in \sett^n $ to an element $F_i(y^n)$ in a smaller set $\{1,\ldots,M\}$. After $y^n \in \sett^n$ has been received by the sender (encoder and estimator) via the noiseless feedback channel, the encoder sends $\boldsymbol{u}_{F_i(y^n))}$, if $i \in \{1,\ldots,N\}$ is given to him. Note that we defined for each coloring function $F_i$ an encoding strategy $\boldsymbol{f}_i=[f_i^1,\ldots,f_i^m] \in \mathcal{F}_m$ as presented in Definition \ref{def:RIDf}.
 If $y^n \notin \sett^n$, an error is declared. This error probability goes to zero as $n$ goes to infinity as computed in \eqref{eq:errorPb}. For a fixed family of maps $\{F_i,\ i=1,\ldots,N\}$ and for each $i \in \{1,\ldots,N\}$, we define the decoding sets $\setd(F_i)=\bigcup_{y^n\in \sett^n}\{y^n\}\times \setd_{F_i(y^n)}$.
 \subsubsection{Error Analysis}
 Now, we analyze the maximal error performance of the deterministic ID feedback code. 
For the analysis of the error of the first kind, we choose a fixed set $\{F_i,\ i=1,\ldots,N\}$. The error of the first kind is upper-bounded by
  \begin{align}
   \mathbb{E}_{S^m}[W^m_{S}({\setd'}_i^c|\boldsymbol{f}_i,S^m)] & 
  = \mathbb{E}_{S^m}[W^m_{S}({{\setd}(F_i)}^c|\boldsymbol{f}_i,S^m)] \nonumber \\
    & \overset{(a)}{\leq}  \sum_{s^m \in \sets^m}  P_S^m(s^m) ( W^n_{S} \big(  (\sett^n)^c | {x^\star}^n, s^n \big)  +  W^{\lceil \sqrt{n} \rceil}_S(\setd_{F_i(y^n)}|\boldsymbol{u}_{F_i(y^n)},s^{\lceil \sqrt{n} \rceil}) ) \nonumber \\
    & \overset{(b)}{\leq} 2^{-n\delta} +2^{-\sqrt{n}\delta^\prime} , 
    \end{align} 
where $(a)$ follows from the memorylessness property of the channel and the union bound and $(b)$ follows from Lemma \ref{lemma:typ1} and the definition of the transmission code $\mathcal{C}$. \\
 In order to achieve a small error of the second kind, we choose suitable maps $\{F_i,\ i=1,\ldots,N\}$ randomly. For $i \in \{1,\ldots,N\}$, $y^n \in \sett^n$ let $\bar{F}_i(y^n)$ be independent RVs such that
 \begin{equation}
    \Pr\{\bar{F}_i(y^n)=j\} = \frac{1}{M}, \quad j \in \{1,\ldots,M\}. \label{eq:mean_of_Fi}
    \end{equation} 
Let $F_1$ be a realization of $\bar{F}_1$. For each $y^n\in \sett^n$, we define the RVs $\psi_{y^n}=\psi_{y^n}(\bar{F}_2)$ analogously to \cite[Section IV]{Idfeedback}.
\begin{equation}
\psi_{y^n}=\psi_{y^n}(\bar{F}_2)=\begin{cases}
1, & \text{if } F_1(y^n)=\bar{F}_2(y^n)\\
0, & \text{otherwise}.
\end{cases}
\end{equation}
The $\psi_{y^n}$ are also independent for every $y^n \in \sett^n$. The expectation of $\psi_{y^n}$ is computed as follows:
\begin{align}
\mathbb{E}[\psi_{y^n}] &  
 = \Pr\{ F_1(y^n)=\bar{F}_2(y^n) \}
 = \frac{1}{M}. \label{eq:mean_Phi}
\end{align}
Since for all $y^n \in \sett^n$, the $\psi_{y^n}$ are i.i.d., we obtain by applying Hoeffding's inequality \ref{HoeffIneq} the following Lemma:
\begin{Lemma} \label{lemmahof}
For $\lambda \in (0,1)$, $\frac{1}{M} < \lambda$, for each channel $V^a_S$ with $\norm{V_S^a-W_S^a} \leq \epsilon$ and for each $n \geq n_0(\epsilon)$
\begin{align}
    \Pr \{ \sum_{y^n \in \sett^n} \psi_{y^n} > |\mathcal{T}_{V_S^a}^n ({x^\star}^n)| \cdot \lambda \}  \leq 2^{-{|\sett^n|}\cdot \lambda\sqrt{n} \epsilon}
\end{align} 
\end{Lemma}
 We derive an upper bound on the error of the second kind for those values of $\bar{F}_2$ satisfying Lemma \ref{lemmahof}. 
 \begin{align}
&\mathbb{E}_{S^m} [W^m_{S}(\setd(\bar{F}_2)|\boldsymbol{f}_1,S^m)] \\
& = \sum_{s^m \in \sets^m} P_S^m(s^m) W^m_{S}(\setd(\bar{F}_2)|\boldsymbol{f}_1,s^m) \nonumber\\
& = \sum_{s^m \in \sets^m} P_S^m(s^m) W^m_{S} \bigg (\setd(\bar{F}_2) \cap \big( (\sett^n \times \sety^{\lceil \sqrt{n} \rceil})
\cup (\sett^n \times \sety^{\lceil \sqrt{n} \rceil})^c \big)|f_1,s^m \bigg) \nonumber\\
&\overset{(a)}{\leq} \sum_{s^m \in \sets^m} P_S^m(s^m) W^m_{S}\big( (\sett^n \times \sety^{\lceil \sqrt{n} \rceil})^c|f_1,s^m \big) \nonumber\\
& \quad + \sum_{s^m \in \sets^m} P_S^m(s^m) W^m_{S}\big( \setd(\bar{F}_2) \cap (\sett^n \times \sety^{\lceil \sqrt{n} \rceil})|f_1,s^m \big) \nonumber \\
 & \overset{(b)}{\leq} \sum_{s^n \in \sets^n} P_S^n(s^n) W^n_{S}\big( (\sett^n)^c|{x^{\star}}^n,s^n \big) \nonumber \\
& + \sum_{s^m \in \sets^m} P_S^m(s^m) \bigg( \sum_{ \substack{ y^n \in \sett^n\\ F_1(y^n) \neq \bar{F}_2(y^n) }} W_S^n(y^n|{x^{\star}}^n,s^n) 
\cdot W_S^{\lceil \sqrt{n} \rceil} ( y^{\lceil \sqrt{n} \rceil}|\boldsymbol{u}_{F_1(y^n)}, s^{\lceil \sqrt{n} \rceil} ) \nonumber \\
& \quad + \sum_{ \substack{ y^n \in \sett^n\\ F_1(y^n) = \bar{F}_2(y^n) }} W_S^n(y^n|{x^{\star}}^n,s^n) \bigg) \nonumber \\
& \leq 2^{-n \delta^\prime} + 2^{-\sqrt{n}\delta} 
+ \sum_{s^n \in \sets^n} P_S^n(s^n)  \sum_{ \substack{ y^n \in \sett^n\\ F_1(y^n) = \bar{F}_2(y^n) }} W_S^n(y^n|{x^{\star}}^n,s^n) \nonumber \\
& \overset{(c)}{\leq} 2^{-n \delta^\prime} + 2^{-\sqrt{n}\delta} \nonumber \\
& \quad + \sum_{V_S^a \colon |V_S^a-W_S^a| \leq \epsilon} ({W_S^a})^n( \mathcal{T}_{V_S^a}^n ({x^\star}^n)|{x^\star}^n) 
 \cdot \big (\sum_{y^n \in \mathcal{T}_{V_S^a}^n ({x^\star}^n) } \psi_{y^n} \big) \cdot |\mathcal{T}_{V_S^a}^n ({x^\star}^n)|^{-1} \nonumber \\
& \overset{(d)}{\leq}  2^{-n \delta^\prime} + 2^{-\sqrt{n}\delta} + \lambda,
\label{eq:error2_Bound}
\end{align}
where $(a)$ follows from the union bound, $(b)$ follows from the memorylessness property of the channel and the union bound, $(c)$ follows from Lemma \ref{lemma:typ1}, from the definition of the transmission code $\mathcal{C}$ and from the definition of the set $\sett^n$ in \eqref{eq:definitionOfTn_d} and $(d)$ follows from Lemma \ref{lemmahof}.

 We repeatedly do the same analysis of the error of the second kind for all pairs $(i_1,i_2) \in \{1,\ldots,N\}^2,\quad i_1 \neq i_2$. 
 \label{app:CisanIDcode}
For notation simplicity, we denote the error of the second kind between the pair $(i_1,i_2)$ by $\mu_2^{(i_1,i_2)}$. We have
 \begin{align}
 \Pr\{\Cp \text{ is not an }(n,N,\lambda_1,\lambda_2) \text{ code} \} &
 = \Pr\{ \bigcup_ {\substack{i_1,i_2 \in \{1,\ldots,N\} \nonumber \\ i_1\neq i_2} \mu_2^{(i_1,i_2)}}   \geq \lambda_2 \}  \nonumber\\
& = \Pr\{ \bigcup_ {\substack{i_1,i_2 \in \{1,\ldots,N\} \nonumber \\ i_1\neq i_2} }  \mu_2^{(i_1,i_2)} \geq \lambda + 2^{- n\delta^\prime}+2^{-\sqrt{n}\delta} \} \nonumber \\
& \overset{(a)}{\leq} N\cdot(N-1) \cdot 2^{-{|\sett^n|}\cdot \lambda\sqrt{n} \epsilon},
\end{align}
 where $(a)$ follows from the union bound, \eqref{eq:error2_Bound} and Lemma \ref{lemmahof}.
 It can be verified that we can construct an ID feedback code for $W_S$ with cardinality $N$ satisfying
 \begin{equation}
     N \geq (n+1)^{-2 |\setx||\sety|} \cdot 2^{|\sett^n| \cdot \lambda\sqrt{n}\epsilon}
 \end{equation}
 and with an error of the second kind upper-bounded as in \eqref{eq:error2_Bound}. 

 The next step in the proof is dedicated to the state estimator. The per symbol distortion defined in \eqref{eq:persymbolDistortion} can be rewritten as follows:
 \begin{align}
    d_t
     &=\mathbb{E}\left[d(S_t,\hat{S}_t)\right]\\
     &=\mathbb{E}\left[\mathbb{E}\left[d(S_t,\hat{S}_t)|X_t,Y_t\right]\right]\\
     &=\sum_{(x,y)\in\setx\times\sety}P_{X_{t}Y_{t}}(x,y)\sum_{s_{t}\in\sets}P_{S_{t}|X_{t}Y_{t}}(s|x,y)\sum_{\hat{s}_{t}\in\sets}P_{\hat{S}_{t}|X_{t}Y_{t}}(\hat{s}|x,y)d(s,\hat{s})\\
     &=\sum_{(x,y)\in\setx\times\sety}P_{X_{t}Y_{t}}(x,y)\sum_{s_{t}\in\sets}P_{S_t|X_tY_t}(s|x,y)d(s,h^\star(x,y))\\
     &=\sum_{x\in\setx}P_{X_t}(x)\mathbb{E}_{S_tY_t}\left[d(S_t,h^\star(X_t,Y_t))|X_t=x\right]\\
     &=\sum_{x\in\setx}P_{X_t}(x)d^\star(x)\\
     &\leq D.
 \end{align}
 This completes the direct proof of Theorem \ref{theorem:capDist1}. \qed
\subsection{Converse Proof of Theorem \ref{theorem:capDist1}}\label{sec:rIDS}
For the converse proof, we use the techniques for deterministic ID for $W_S$ as described in \ref{sec_converse_IDS}. We first extend the Lemma \ref{lemma3} to the joint deterministic ID and sensing problem. 
\begin{Lemma} For any feedback strategy $\boldsymbol{f}_{D}=\left[f_{D}^1,f_{D}^2,\cdots,f_{D}^n\right]$ which satisfies the \emph{per symbol distortion} constraint as described in \eqref{eq:minimalDis1}, i.e., for all $t\in\left\{1,2,\cdots,n\right\}$, $d^{\star}(f_{D}^t)\leq D$, and for any $\mu\in\left(0,1\right)$, we have
\begin{align}
    \min_{\sete_3\subset\sety^n:\mathbb{E}_{S^n}\left[W^n_S(\sete_3|\boldsymbol{f}_D,S^n)\right]\geq 1-\mu}\left|\sete_3\right|\leq K_3,
\end{align}
where $K_3$ is given by
\begin{align}
    K_3=2^{n\max_{x\in\setx_D}H(\mathbb{E}\left[W_S(\cdot|x)\right])+\alpha\sqrt{n}}=2^{nH(\mathbb{E}\left[W_S(\cdot|x^{\star})\right])+\alpha\sqrt{n}},
\end{align}
where
\begin{align}
    \alpha=\sqrt{\frac{\beta}{\mu}},\quad \beta=\max\left(\log^2(3),\log^2\left(\left|\sety\right|\right)\right).
\end{align}
\label{lemma:K3}
\end{Lemma}
\begin{proof}
    Let $\sete_3^{\star}\subset\sety^n$ be defined as follows:
    \begin{align}
        \sete_3^{\star}=\left\{y^n\in\sety^n,-\log{\mathbb{E}_{S^n}\left[W_S^n\left(\boldsymbol{f}_D,S^n\right)\right]\leq K_3}\right\}.
    \end{align}
    Define a RV $Z_t=-\log{\mathbb{E}_{S_t}\left[W_S(Y_t|f^t_D(Y^{t-1}),S_t)\right]}$.
By \eqref{eq:Z_t} we have
\begin{align}
    \mathbb{E}_{S^n}\left[W_S^n\left(\sete^{\star}_3|\boldsymbol{f}_{D},S^n\right)\right]=Pr\left\{\sum_{t=1}^n Z_t\leq \log{K_3}\right\}.\label{eq:didas1}
\end{align}
Similarly, for all $t\in\left\{1,2,\cdots,n\right\}$, we define RV $U_t=Z_t-\mathbb{E}\left[Z_t|Y^{t-1}\right]$. It has been shown in \eqref{eq:u_t1} and \eqref{eq:u_t2} that $\mathbb{E}\left[U_t|Y^{t-1}\right]=0$ and $\mathbb{E}\left[U_t\right]=0$. 

Moreover, for all $t\in\left\{1,2,\cdots,n\right\}$ and for all $y^{t-1}\in\sety^{t-1}$, we have
\begin{align}
    \mathbb{E}\left[Z_t|y^{t-1}\right]
    &=\sum_{y_t\in\sety}\left(-\mathbb{E}_{S_t}\left[W_S(Y_t|f^t_D(Y^{t-1}),S_t)\right]\log{\mathbb{E}_{S_t}\left[W_S(Y_t|f^t_D(Y^{t-1}),S_t)\right]}\right)\nonumber\\
    &\leq \max_{x\in\setx_D}H\left(\mathbb{E}_S\left[W_S(\cdot|x,S)\right]\right)\nonumber\\
    &= H\left(\mathbb{E}_{S}\left[W_S(\cdot|x^{\star},S)\right]\right).\label{eq:didas2}
\end{align}
By combing \eqref{eq:didas1} and \eqref{eq:didas2} we have
\begin{align}
    \mathbb{E}_{S^n}\left[W_S^n\left(\boldsymbol{f}_{D},S^n\right)\right]
    &\geq Pr\left\{\sum_{t=1}^n{U_t}\leq \alpha\sqrt{n}\right\}\nonumber\\
    &\geq 1-\mu.  
\end{align}
This completes the proof of Lemma \ref{lemma:K3}
\end{proof}
The subsequent proof steps are identical to Section \ref{sec_converse_IDS}.
\subsection{Direct Proof of Theorem \ref{theorem:capDist2}}
For the direct proof of this theorem, we follow a code construction similar to that presented in \cite{Idfeedback}, with one key difference: here, we optimize only over input distributions that satisfy the per-symbol constraint $d^\star(P)\leq D$. 
\subsubsection{Coding Scheme}
We construct a randomized ID code  with blocklength $m=n+\sqrt{n}$ by concatenating two transmission codes, which will be described in details subsequently. The first $n$ symbols are allocated for generation of common randomness. We employ distribution 
\[
P^*=\arg\max_{P\in\setp(\setx)}H(\sum_{x\in\setx}P(x)\mathbb{E}\left[W_S(\cdot|x,S))\right],\] 
where the maximization is performed over distributions subject to the constraint $\setp_{D}=\left\{P\in\setp(\setx)|d^{\star}(P)\le D\right\}$. Regardless of which identity $i \in \{1,\ldots,N\}$ we want to identify, the sender first sends $n$ symbols with respect to the distribution ${P^\star}^n\in \setp_D^n$ over the state-dependent channel $W_S$. The received sequence $y^n \in \sety^n$ becomes known to the sender (estimator and encoder) via the noiseless feedback link. The feedback provides the sender and the receiver with the knowledge of the outcome of the correlated random experiment $\bigg( \sety^n, \sum_{x^n\in\setx^n}P^{\star n}(x^n)\mathbb{E}_{S^n}\big[W_S^n(\cdot|x^n,S^n) \big] \bigg)$. 

\subsubsection{Common Randomness Generation}
Similar as the deterministic coding scheme, we want to generate uniform common randomness. Therefore, we convert our correlated random experiment $$\bigg( \sety^n, \sum_{x^n\in\setx^n}P^{\star n}(x^n)\mathbb{E}_{S^n}\big[W_S^n(\cdot|{x^\star}^n,S^n) \big] \bigg)$$ to a uniform one $\bigg( \mathcal{T}'^n, \sum_{x^n\in\setx^n}P^{\star n}(x^n)\mathbb{E}_{S^n}\big[W_S^n(\cdot|{x^\star}^n,S^n) \big] \bigg) \bigg)$. For $\epsilon >0$, the set $\mathcal{T}'^n$ is given by
 \begin{equation}
     \sett'^n= \bigcup_{V_S^a \colon \norm{V_S^a-W_S^a}\leq \epsilon} \sum_{x\in\setx}P^{\star n}(x^n)\mathcal{T}_{V_S^a}^n (x^n), \label{eq:definitionOfTn_r}\end{equation}

By Lemma \ref{lemma:typ2}, we can get the following corollary:

\begin{Corollary}
 Let $(x^n,y^n)$ be emitted by the DMS $P_{XY}(\cdot,\cdot)= W(\cdot|\cdot)P_X(\cdot)$ and let $V \in \mathcal{W}$. For every $\epsilon >0$ there exist a $c(\epsilon)>0$ and $n_0(\epsilon)$ such that for $n\geq n_0(\epsilon)$
\begin{enumerate}
\item $ | \underset{V \colon \norm{V-W} \leq \epsilon}{\bigcup} \sett^n_V(x^n) | \geq 2^{n (H(\sum_{x\in\setx}P_X(x)W(\cdot|x)) - c(\epsilon) )} $
\item $ | \underset{V \colon \norm{V-W} \leq \epsilon}{\bigcup} \sett^n_V(x^n) | \leq 2^{n (H(\sum_{x\in\setx}P_X(x)W(\cdot|x)) +c(\epsilon) )} $
\item if $\norm{V-W}\leq \epsilon$, $\sett_V^n(x^n) \neq \emptyset$ and $c(\epsilon) \to 0$ if $\epsilon \to 0$, then
\begin{equation}
    |\sett_V^n(x^n)| \geq 2^{n (H(\sum_{x\in\setx}P_X(x)W(\cdot|x)) - c(\epsilon) )}.
\end{equation}
\end{enumerate}
\label{lema:typ2r}
\end{Corollary}
Therefore, we have $\left|\sett'^n\right|\approx 2^{nH\left(\sum_{x\in\setx}{P(x)\mathbb{E}_S\left[W(\cdot|x,S)\right]}\right)}$. The sender generate randomness according to the random experiment $\left(\sett'^n,\sum_{x^n\in\setx^n}P^{\star n}(x^n)\mathbb{E}_{S^n}\left[W^n_S(\cdot|x^n,S^n)\right]\right)$. Asymptotically, the error probability of $y^n\notin \sett'^n$ goes to zero. Similar as the deterministic scheme, we prepare the coloring functions $\left\{F_i, i=1,\cdots,N\right\}$. The last $\lceil\sqrt{n}\rceil$ symbols are used to transmit $F_i(y^n)$ using a standard $\left(\lceil\sqrt{n}\rceil,M,2^{-n\sqrt{n}\delta}\right)$ transmission code $C=\left\{\left(\boldsymbol{u_j},\setd_j\right),j=1,\cdots,M\right\}$, where $\boldsymbol{u}_j\in\setx_D^{\lceil\sqrt{n}\rceil}$ for each $j=1,\cdots,M$. The probability distribution for encoding is defined as $Q_F(\cdot|i)={P^{\star}}^n\times\mathbb{I}\left\{\boldsymbol{x_n^m}=F_i(y^n)\right\}$ and decoding region is given by $\setd(F_i)=\bigcup_{y^n\in\sett'}\left\{y^n\right\}\times\setd(F_i(y^n))$. 

\subsubsection{Error Analysis}
Following that, for all $i=1,\cdots,N$, the error of the first kind $P_{e,1}(i)$ is upper-bounded by
\begin{align}
    P_{e,1}(i)
    &=\sum_{\boldsymbol{f}\in\setf^m}Q_F(\boldsymbol{f}|i)\mathbb{E}_{S^m}\left[W_S^m(\setd(F_i)|\boldsymbol{f},S^m)\right]\\
    &=\sum_{s^m\in\sets^m}P_S^m(s^m)\sum_{\boldsymbol{f}\in\setf^m}Q_F(\boldsymbol{f}|i)W_S^m\left(\left(\bigcup_{y^n\in\sett'}y^n\times \setd_{F_i(y^n)}\right)^c|\boldsymbol{f},S^m\right)\\
    &\overset{(a)}{\leq} \sum_{s^m\in\sets^m}P_S^m(s^m)\left(\sum_{x^n\in\setx^n}{P^{\star}}^n(x^n)W_S^n(\sett'^c|x^n,S^n)+W_S^{\lceil\sqrt{n}\rceil}\left(\setd_{F_i(y^n)}^c|\boldsymbol{u}_{F_i(y^n)},S^{\lceil\sqrt{n}\rceil}\right)\right)\\
    &\overset{(b)}{\leq} 2^{-n\delta}+2^{-\sqrt{n}\delta'},
\end{align}
where $(a)$ follows from the union bound and $(b)$ follows from Corollary \ref{lema:typ2r}.
Furthermore, for all $i,j=1,\cdots,N$ with $i\neq j$, the probability of error of the second kind $P_{e,2}(i,j)$ should be asymptotically upper-bounded by $\lambda$. Without the loss of generalization, we fix $i=1$, $j=2$ and examine the error probability $P_{e,2}(1,2)$
 \begin{align}
 P_{e,2}(1,2)
 &=\sum_{\boldsymbol{f}\in\setf^m}Q(\boldsymbol{f}|i=1)\mathbb{E}_{S^m} [W^m_{S}(\setd(\bar{F}_2)|\boldsymbol{f},S^m)] \nonumber\\
& = \sum_{s^m \in \sets^m} P_S^m(s^m) \sum_{\boldsymbol{f}\in\setf^m} Q(\boldsymbol{f}|i=1) W^m_{S}(\setd(\bar{F}_2)|\boldsymbol{f},s^m) \nonumber\\
& = \sum_{s^m \in \sets^m}  P_S^m(s^m)\sum_{\boldsymbol{f}\in\setf^m}Q(\boldsymbol{f}|i=1)\cdot\nonumber\\
&\quad \cdot  W^m_{S} \bigg (\setd(\bar{F}_2) \cap \big( (\sett'^n \times \sety^{\lceil \sqrt{n} \rceil})
\cup (\sett'^n \times \sety^{\lceil \sqrt{n} \rceil})^c \big)|\boldsymbol{f},s^m \bigg) \nonumber\\
&\overset{(a)}{\leq} \sum_{s^m \in \sets^m} P_S^m(s^m) \sum_{\boldsymbol{f}\in\setf^m}Q(\boldsymbol{f}|i=1)W^m_{S}\big( (\sett^n \times \sety^{\lceil \sqrt{n} \rceil})^c|\boldsymbol{f},s^m \big) \nonumber\\
& \quad + \sum_{s^m \in \sets^m} P_S^m(s^m)\sum_{\boldsymbol{f}\in\setf^m} Q(\boldsymbol{f}|i=1) W^m_{S}\big( \setd(\bar{F}_2) \cap (\sett^n \times \sety^{\lceil \sqrt{n} \rceil})|\boldsymbol{f},s^m \big) \nonumber \\
 & \overset{(b)}{\leq} \sum_{s^m \in \sets^m} P_S^m(s^m)\sum_{x^n\in\setx^n}{P^{\star}}^n\left(x^n\right) W^n_{S}\big( (\sett^n)^c|x^n,s^n \big)W^{\lceil\sqrt{n}\rceil}_S\left(\setd_{F_i(y^n)}^c|\boldsymbol{u}_{F_i(y^n)},s^{\lceil\sqrt{n}\rceil}\right) \nonumber \\
& + \sum_{s^m \in \sets^m} P_S^m(s^m) \bigg( \sum_{ \substack{ y^n \in \sett^n\\ F_1(y^n) \neq \bar{F}_2(y^n) }} \sum_{x^n\in\setx^n}{P^{\star}}^n(x^n) W_S^n(y^n|x^n,s^n) \cdot
 \nonumber \\
& \quad \cdot W_S^{\lceil \sqrt{n} \rceil} ( y^{\lceil \sqrt{n} \rceil}|\boldsymbol{u}_{F_1(y^n)}, s^{\lceil \sqrt{n} \rceil} )+ \sum_{ \substack{ y^n \in \sett^n\\ F_1(y^n) = \bar{F}_2(y^n) }}\sum_{x^n\in\setx^n}{P^{\star}}^n(x^n) W_S^n(y^n|x^n,s^n) \bigg) \nonumber \\
& \leq 2^{-n \delta^\prime} + 2^{-\sqrt{n}\delta} 
+ \sum_{s^n \in \sets^n} P_S^n(s^n)  \sum_{ \substack{ y^n \in \sett^n\\ F_1(y^n) = \bar{F}_2(y^n) }} \sum_{x^n\in\setx^n}{P^{\star}}^n(x^n)W_S^n(y^n|x^n,s^n) \nonumber \\
& \overset{(c)}{\leq} 2^{-n \delta^\prime} + 2^ + \sum_{x^n\in\setx^n}{P^{\star}}^n(x^n)\sum_{V_S^a \colon |V_S^a-W_S^a| \leq \epsilon} ({W_S^a})^n( \sett{'}_{V_S^a}^n (x^n)|x^n) \cdot \nonumber \\
& \quad 
 \cdot \big(\sum_{y^n \in \mathcal{T}{'}_{V_S^a}^n (x^n) } \psi_{y^n} \big) \cdot |\mathcal{T}{'}_{V_S^a}^n (x^n)|^{-1} \nonumber \\
& \overset{(d)}{\leq}  2^{-n \delta^\prime} + 2^{-\sqrt{n}\delta} + \lambda,
\end{align}
where $(a)$ follows from the union bound, $(b)$ follows from the memeoryless channel and the union bound, $(c)$ follows from Lemma \ref{lemma:typ1}, from the definition of the transmission code $\mathcal{C}$ and from the definition of the set $\sett^n$ in \eqref{eq:definitionOfTn_d} and $(d)$ follows from Corollary \ref{lema:typ2r}.\\
 We repeatedly do the same analysis of the error of the second kind for all pairs $(i_1,i_2) \in \{1,\ldots,N\}^2,\quad i_1 \neq i_2$. We can verify that we can construct a randomized ID feedback code for $W_S$ with cardinality $N$ satisfying
 \begin{equation}
     N \geq (n+1)^{-2 |\setx||\sety|} \cdot 2^{|\sett'^n| \cdot \lambda\sqrt{n}\epsilon}
 \end{equation}
 and with an error of the first kind and second kind upper-bounded as in \eqref{eq:error1r} and \eqref{eq:error2r}, respectively. 

 Finally the state estimator is checked. The per symbol distortion defined in \eqref{eq:persymbolDistortion} can be rewritten as follows:
 \begin{align}
    d_t
     &=\mathbb{E}\left[d(S_t,\hat{S}_t)\right]\\
     &=\mathbb{E}\left[\mathbb{E}\left[d(S_t,\hat{S}_t)|X_t,Y_t\right]\right]\\
     &=\sum_{(x,y)\in\setx\times\sety}P_{X_{t}Y_{t}}(x,y)\sum_{s_{t}\in\sets}P_{S_{t}|X_{t}Y_{t}}(s|x,y)\sum_{\hat{s}_{t}\in\sets}P_{\hat{S}_{t}|X_{t}Y_{t}}(\hat{s}|x,y)d(s,\hat{s})\\
     &=\sum_{(x,y)\in\setx\times\sety}P_{X_{t}Y_{t}}(x,y)\sum_{s_{t}\in\sets}P_{S_t|X_tY_t}(s|x,y)d(s,h^\star(x,y))\\
     &=\sum_{x\in\setx}P_{X_t}(x)\mathbb{E}_{S_tY_t}\left[d(S_t,h^\star(X_t,Y_t))|X_t=x\right]\\
     &=d^{\star}(P)\\
     &\leq D.
 \end{align}
 This completes the direct proof of Theorem \ref{theorem:capDist2}. \qed

\subsection{Converse Proof of \ref{theorem:capDist2}}
We first extend the Lemma \ref{lemma:K2} to the joint randomized ID and sensing problem.
\begin{Lemma}
    For any randomized feedback strategy $Q_{D}(\boldsymbol{f})=\prod_{t=1}^n Q^{t}_{D}(f^t)$ over all n-length feedback encoding set $\setf_n$, which satisfies the per symbol distortion constraint as described in \eqref{eq:didas2}, i.e., for all $t\in\{1,2,\cdots,n\}$, $d^{\star}\left(Q_{D}^{t}(f^t)\right)\leq D$, and for any $\mu\in(0,1)$, 
    \begin{align}
        \min_{\sete_4\subset\sety^n:\sum_{\boldsymbol{f}\in\setf_n}Q_D(\boldsymbol{f})\mathbb{E}_{S^n}\left[W^n_S(\sete_4|\boldsymbol{f},S^n)\right]\geq 1-\mu}|\sete_4|\leq K_4,
    \end{align}
    where $K_{4}$ is given by 
    \begin{align}
        K_{4}=2^{n\max_{P\in\setp^D}H\left(\sum_{x\in\setx}P(x)\mathbb{E}\left[W_S\left(\cdot|x,S\right)\right]\right)+\alpha\sqrt{n}}=2^{nH\left(\sum_{x\in\setx}P^{\star}(x)\mathbb{E}\left[W_S(\cdot|x,S)\right]\right)+\alpha\sqrt{n}},
    \end{align}
    where $\alpha=\sqrt{\frac{\beta}{\mu}}$, $\beta=\max{\log^2{3},\log^2{|\sety|}}$.
    \label{lemma:K4}
    \end{Lemma}

    \begin{proof}
        Define a set $\sete^{\star}_4\subset$ as follows:
        \begin{align}
        \sete_{4}^{\star}=\left\{y^n\in\sety^{n}, -\log{\sum_{\boldsymbol{f}\in\setf_n}Q_D\left(\boldsymbol{f}\right)\mathbb{E}_{S^n}\left[W_S(y^n|\boldsymbol{f},S^n)\right]}\leq \log{K_4}\right\}.
    \end{align}
    Define an auxiliary RV $Z_t$ as follows:
    \begin{equation}
     Z_t=-\log\left(\sum_{f^t\in\setf^t}Q_{D}^{t}(f^t)\mathbb{E}_{S_t}\left[W^n_S(Y_t|f^t(Y^{t-1}),S_t)\right]\right).
 \end{equation}
By \eqref{eq:rid1} we have
\begin{align}
    \sum_{\boldsymbol{f}\in\setf_n}Q_D(\boldsymbol{f})\mathbb{E}\left[W_S^n\left(\sete_4|\boldsymbol{f},S^n\right)\right]=Pr\left\{\sum_{t=1}^nZ_t\leq \log{K_4}\right\}.
\end{align}
Similarly, for any $y^{t-1}\in\sety^{t-1}$, we examine
\begin{align}\nonumber
    &\mathbb{E}\left[Z_t|y^{t-1}\right]\\
    &=\sum_{y_t\in\sety}(-\sum_{f^t\in\setf^t}Q_D^{t}(f^t)\mathbb{E}_{S_t}[W^n_S(y_t|f^t(y^{t-1}),S_t)]\log(\sum_{f^t\in\setf^t}Q_D^{t}(f^t)\mathbb{E}_{S_t}[W^n_S(y_t|f^t(y^{t-1}),S_t)]))\nonumber\\
    &\leq \max_{P\in\setp_D}H\left(\sum_{x\in\setx}P(x)\mathbb{E}\left[W_S(\cdot|x,S)\right]\right)\nonumber\\
    &= H\left(\sum_{x\in\setx}P^{\star}(x)\mathbb{E}\left[W_S^n(\cdot|x,S\right]\right).
    \label{eq:jridas2}
\end{align}
Applying Chebyshev's inequality, we complete the proof of Lemma \ref{lemma:K4}.
\end{proof}
The proof steps that follow are the same as those in Section \ref{sec:rIDS}.

\section{Average Distortion}\label{average}
In addition to the per-symbol distortion constraint, an alternative and more flexible distortion constraint is average distortion. This approach is valuable because it relaxes the per-symbol fidelity requirement, allowing for minor variations in individual symbol quality as long as the overall average distortion remains below a specified threshold. As defined in \cite{labidi2023joint}, the average distortion for a sequence of symbols is given by
\begin{align}
    \bar{d}^n
    &=\mathbb{E}_{S^n\hat{S}^n}\left[d(S^n,\hat{S}^n)\right]\nonumber\\
    &=\frac{1}{n}\sum_{t=1}^{n}\mathbb{E}_{S_t,\hat{S}_t}\left[d(S_t,\hat{S}_t)\right].
\end{align}
 This metric captures the average quality of the reconstructed sequence, making it suitable for applications where consistent, strict fidelity for each symbol is not essential but the overall fidelity of the transmission needs to remain within acceptable limits.

In the case of a deterministic ID code, the average distortion can be expressed in a more detailed form:
\begin{equation}
    \bar{d}^n=\frac{1}{N} \sum_{i =1}^{N} \mathbb{E}_{S^nY^n} \big[\frac{1}{n} \sum_{t=1}^{n} d(S_t,\hat{S}_t)|X^n=\boldsymbol{f}_i\big].\label{eq:averageDistortion}
\end{equation}
 Using the code construction method from Section \ref{sec: direct JdIDAS} along with the minimum distortion condition defined in \eqref{eq:minimalDis1}, we propose the following theorem, which provides a lower bound on the deterministic ID capacity-distortion function for a state-dependent channel $W_S$ under an average distortion constraint:
\begin{Theorem}
    The deterministic ID capacity-distortion function with average distortion constraint $\bar{d}^n\leq D$ of the state-dependent channel $W_S$ is lower-bounded as follows:
\begin{equation}
    C_{ID,avg}^d (D) \geq \max_{x \in \setx_D}  H\big(\mathbb{E}[W_S(\cdot|x,S)] \big),
\end{equation}
where the set $\setx_D$ is given by
\begin{equation}
    \setx_D=\{x \in \setx,\quad d^\star(x) \leq D\}.
\end{equation}
\label{theorem:capDist_avg}
\end{Theorem}
Despite the practical implications of this result, proving a converse theorem for this bound remains an open problem.

 \section{Conclusion and Discussion}
 In our work, we studied the problem of joint ID and channel state estimation over a DMC with i.i.d. state sequences, where the sender simultaneously sends an identity and senses the state via a strictly causal channel output. We established the capacity on the corresponding ID capacity-distortion function. It emerges that sensing can increase the ID capacity. In the proof of our theorem, we notice that the generation of common randomness is a key tool to achieve a high ID rate. The common randomness generation is helped by feedback. The ID rate can be further increased by adding local randomness at the sender.

Our framework closely mirrors the one described in \cite{MariPaper}, with the key distinction being that we utilize an identification scheme instead of the classical transmission scheme. We want to simultaneously identify the message sent and estimate the channel's state. As noted in the results of \cite{MariPaper}, the capacity-distortion function is consistently smaller than the transmission capacity of the state-dependent DMC, except when the distortion is infinite. This observation aligns with expectations for the message transmission scheme, as the optimization is performed over a constrained input set defined by the distortion function.
However, this does not directly apply to the ID scheme. An interesting aspect is that the capacity-distortion function for the deterministic encoding case scales double exponentially with the blocklength, as highlighted in Theorem \ref{theorem:capDist1}. However, the ID capacity of the state-dependent DMC with deterministic encoding scales only exponentially with the blocklength. This is because feedback significantly enhances the ID capacity, enabling a double-exponential growth of the ID capacity for the state-dependent DMC, as established in Theorem \ref{MainTheorem1}. This contrasts sharply with the message transmission scheme, where feedback does not increase the capacity of a DMC. Introducing an estimator into our framework naturally reduces the ID capacity compared to the scenario with feedback but without an estimator. This reduction occurs because the optimization is performed over a constrained input set defined by the distortion function. Nevertheless, the capacity-distortion function remains higher than in the case without feedback and without an estimator.
This difference underscores a unique characteristic of the ID scheme, highlighting its distinct scaling behavior and potential advantages in certain scenarios.

We consider two cases: deterministic and randomized identification. For a transmission system without sensing, it was shown in \cite{AhlDueck,deterministicDMC} that the number of messages grows exponentially, i.e., $N = 2^{nC_{ID}^d}$. 

Remarkably, Theorem~\ref{theorem:capDist1} demonstrates that by incorporating sensing, the growth rate of the number of messages becomes double exponential: $N = 2^{2^{nC_{ID}^d(D)}}$. This result is notable and closely parallels the findings on identification with feedback \cite{Idfeedback}.

In the case of randomized identification, Theorem~\ref{theorem:capDist2} shows that the capacity is also improved by incorporating sensing. However, in both the deterministic and randomized settings, the scaling remains double exponential.

One application of message identification is in control and alarm systems \cite{6Gcomm,6G_Book}. For instance, it has been shown that identification codes can be used for status monitoring in digital twins \cite{caspar}. Our results demonstrate that, in this context, incorporating a sensing process can significantly enhance the capacity.


Another potential application of our framework is molecular communication, where nanomachines use identification codes to determine when to perform specific actions, such as drug delivery \cite{Labidi2024}. In this context, sensing the position of the nanomachines can enhance the communication rate. 
For such scenarios, it is also essential to explore alternative channel models, like the Poisson channel.

Furthermore, it is clear that in other applications one have to consider different distortion functions. 

In the future, it would be interesting to apply the method to other distortion functions. Furthermore, in practical scenarios, there are models where the receiver either additionally or exclusively performs the sensing. This suggests the need to study various system models. For wireless communications, the Gaussian channel is more practical and widely applicable. Therefore, it would be also valuable to extend our results to the Gaussian case (JIDAS scheme with a Gaussian channel as the forward channel). It has been shown in \cite{isit_paper, identificationwithfeedback} that the ID capacity of the Gaussian channel with noiseless feedback is infinite. Interestingly, the ID capacity of the Gaussian channel with noiseless feedback remains infinite regardless of the scaling used for the rate, e.g., double exponential, triple exponential, etc. By introducing an estimator, we conjecture that the same results will hold, leading to infinite capacity-distortion function.
Considering scenarios with noisy feedback is more practical for future research.

 \label{sec:conclusions}


\section*{Acknowledgement}
The authors acknowledge the financial support by the Federal Ministry of Education and Research
of Germany (BMBF) in the programme of “Souverän. Digital. Vernetzt.”. Joint project 6G-life, project identification number: 16KISK002.
H. Boche and W. Labidi were further supported in part by the BMBF within the national initiative on Post Shannon Communication (NewCom) under
Grant 16KIS1003K and within the national initiative on moleculare communication (IoBNT) under the grant 16KIS1988. C.\ Deppe was further supported in part by the BMBF within the national initiative on Post Shannon Communication (NewCom) under Grant 16KIS1005. C. Deppe, W. Labidi and Y. Zhao were also supported by the DFG within the projects DE1915/2-1 and BO 1734/38-1. This work has been presented in part at the IEEE International Symposium on Information Theory 2023 (ISIT 2023) \cite{labidi2023joint}
\bibliographystyle{IEEEtran}
\bibliography{references}
\IEEEtriggeratref{4}

\end{document}

%% file: figures/System.tex
\scalebox{1.0}{
\tikzstyle{farbverlauf} = [ top color=white, bottom color=white!80!gray]
\tikzstyle{block} = [draw,top color=white, bottom color=white!80!gray, rectangle, rounded corners,
minimum height=2em, minimum width=2.5em]
\tikzstyle{input} = [coordinate]
\tikzstyle{sum} = [draw, circle,inner sep=0pt, minimum size=2mm,  thick]
\scalebox{1}{
\tikzstyle{arrow}=[draw,->] 
\begin{tikzpicture}[auto, node distance=2cm,>=latex']
\node[] (M) {$i \in \mathcal{N}$};
\node[block,right=.5cm of M] (enc) {Encoder};
\node[block, right=1.7cm of enc] (channel) {DMC};
\node[block,below=.7cm of channel](state){$P_S$};
\node[block, right=1cm of channel] (dec) {Decoder};
\node[right=.5cm of dec] (Mhat) {\begin{tabular}{c} Is ${i}^\prime$ sent? \\ Yes or No? \end{tabular}};
\node[input,right=.5cm of channel] (t1) {};
\node[input,above=1cm of t1] (t2) {};
\draw[-{Latex[length=1.5mm, width=1.5mm]},thick] (M) -- (enc);
\draw[-{Latex[length=1.5mm, width=1.5mm]},thick] (enc) --node[above]{ $X^n$} (channel);
\draw[-{Latex[length=1.5mm, width=1.5mm]},thick] (channel) --node[above]{$Y^n$} (dec);
\draw[-{Latex[length=1.5mm, width=1.5mm]},thick] (dec) -- (Mhat);
\draw[-{Latex[length=1.5mm, width=1.5mm]},thick] (state)--(channel);
\end{tikzpicture}}
}

%% file: figures/SystemFeedback.tex
\tikzstyle{farbverlauf} = [ top color=white, bottom color=white!80!gray]
\tikzstyle{block} = [draw,top color=white, bottom color=white!80!gray, rectangle, rounded corners,
minimum height=2em, minimum width=2.5em]
\tikzstyle{input} = [coordinate]
\tikzstyle{sum} = [draw, circle,inner sep=0pt, minimum size=2mm,  thick]
\scalebox{1}{
\tikzstyle{arrow}=[draw,->] 
\begin{tikzpicture}[auto, node distance=2cm,>=latex']
\node[] (M) {$i \in \mathcal{N}$};
\node[block,right=.5cm of M] (enc) {Encoder};
\node[block, right=2.7cm of enc] (channel) {DMC};
\node[block,below=.7cm of channel](state){$P_S$};
\node[block, right=.8cm of channel] (dec) {Decoder};
\node[right=.5cm of dec] (Mhat) {\begin{tabular}{c} Is ${i}^\prime$ sent? \\ Yes or No? \end{tabular}};
\node[input,right=.5cm of channel] (t1) {};
\node[input,above=1cm of t1] (t2) {};
\draw[-{Latex[length=1.5mm, width=1.5mm]},thick] (M) -- (enc);
\draw[-{Latex[length=1.5mm, width=1.5mm]},thick] (enc) --node[above]{ $X_t=f_i^t(Y^{t-1})$} (channel);
\draw[-{Latex[length=1.5mm, width=1.5mm]},thick] (channel) --node[below]{$Y_t$} (dec);
\draw[-{Latex[length=1.5mm, width=1.5mm]},thick] (dec) -- (Mhat);
\draw[-] (t1) -- (t2);
\draw[-{Latex[length=1.5mm, width=1.5mm]},thick] (t2) -| (enc);
\draw[-{Latex[length=1.5mm, width=1.5mm]},thick] (state)--(channel);
\end{tikzpicture}}

%% file: figures/JointSensing.tex
\tikzstyle{farbverlauf} = [ top color=white, bottom color=white!80!gray]
\tikzstyle{block} = [draw,top color=white, bottom color=white!80!gray, rectangle, rounded corners,
minimum height=2em, minimum width=3em]
\tikzstyle{block1} = [draw, fill=none, rectangle, rounded corners,
minimum height=8em, minimum width=2
cm]
\tikzstyle{input} = [coordinate]
\tikzstyle{sum} = [draw, circle,inner sep=0pt, minimum size=2mm,  thick]

\scalebox{1}{
\tikzstyle{arrow}=[draw,->] 
\begin{tikzpicture}[auto, node distance=2cm,>=latex']
\node[] (M) {$i \in \mathcal{N}$};

\node[block,right=.5cm of M] (enc) {Encoder};
\node[block, above=.7cm of enc] (est) {Estimator};
\node[left=.4cm of est] (S) {$\hat{S}^n$};
\node[block, right=3.2cm of enc] (channel) {DMC};
\node[block1, dashed] at (1.75,.75) (tr) {};
\node[above=.3cm of tr] (tran) {sender};
\node[block,below=.7cm of channel](state){$P_S$};
\node[block, right=1cm of channel] (dec) {Decoder};
\node[right=.4cm of dec] (Mhat) {\begin{tabular}{c} {\small Is ${i}^\prime$ sent?} \\ {\small Yes or No?} \end{tabular}};
\node[input,right=.7cm of channel] (t1) {};
\node[above=2cm of channel] (t4) {};
\node[below=1.5cm of t1] (t5) {};
\node[input,above=2.3cm of t1] (t2) {};
\node[input,left=4cm of t2] (point) {};
\node[input,below=1.65cm of point] (ttpoint) {};
\node[input,left=4.3cm of t2] (tt) {};
\node[input,below=1.65cm of tt] (ttt) {};
\node[input,left=0.4cm of ttt] (ttn) {};
\node[input, right=.5cm of enc] (t3) {};
\node[input, below=0.35 cm of est] (tte) {};
\draw[-{Latex[length=1.5mm, width=1.5mm]},thick] (M) -- (enc);
\draw[-{Latex[length=1.5mm, width=1.5mm]},thick] (enc) --node[above]{ $X_t=f_i^t(Y^{t-1})$} (channel);
\draw[-{Latex[length=1.5mm, width=1.5mm]},thick] (channel) --node[below]{$Y_t$} (dec);
\draw[-{Latex[length=1.5mm, width=1.5mm]},thick] (dec) -- (Mhat);
\draw[-] (t1) -- (t2);
\draw[-{Latex[length=1.5mm, width=1.5mm]},thick] (t2) -|  (est);
\draw[-{Latex[length=1.5mm, width=1.5mm]},thick] (state)--(channel);
\draw[-{Latex[length=1.5mm, width=1.5mm]},thick] (est) --(S);
\draw[-{Latex[length=1.5mm, width=1.5mm]},thick] (t3) |- (est);
\draw (ttt) arc[start angle=0, end angle=180, radius=0.2]  (ttn); 
\draw[-{Latex[length=1.5mm, width=1.5mm]},thick] (ttn) -| (enc);
\draw[-] (point) -- (ttpoint);
\draw[-] (ttpoint) -- (ttt);
\end{tikzpicture}}